\documentclass[sigplan,10pt,nonacm]{acmart}
\acmConference[PLDI'22]{ACM SIGPLAN Conference on Programming Language Design
  and Implementation}{June 20--24, 2022}{San Diego, CA, USA}
\acmYear{2022}
\acmISBN{} 
\acmDOI{} 

\setcopyright{none}

\usepackage{booktabs}   
\usepackage{subcaption} 

\newif\ifdraft\draftfalse
\newif\ifappendix\appendixtrue
\newif\ifdesperateforspace\desperateforspacetrue
\newif\ifaftersubmission\aftersubmissionfalse

\definecolor{dkred}{rgb}{0.7,0,0}
\definecolor{dkpurple}{HTML}{4e02eb}
\definecolor{dkgreen}{HTML}{006329}
\definecolor{teal}{HTML}{007982}
\definecolor{string}{HTML}{02782d}
\definecolor{Fuchsia}{HTML}{8C368C}
\definecolor{knowledgelink}{HTML}{000077}

\newcommand{\comm}[3]{\ifdraft\textcolor{#1}{[#2: #3]}\fi}

\newcommand{\bcp}[1]{\comm{dkpurple}{BCP}{#1}}
\newcommand{\hg}[1]{\comm{dkred}{HG}{#1}}

\newcommand{\cc}[1]{{\color{string} \tt #1}}

\newcommand{\bench}[1]{\textbf{#1}}

\newenvironment{proofsketch}{%
  \renewcommand{\proofname}{Proof sketch}\proof}{\endproof}

\theoremstyle{remark}

\theoremstyle{theorem}
\newtheorem{manualtheoreminner}{Theorem}
\newenvironment{repeat-theorem}[1]{%
  \renewcommand\themanualtheoreminner{#1}%
  \manualtheoreminner{}
}{\endmanualtheoreminner}

\theoremstyle{theorem}
\newtheorem{manuallemmainner}{Lemma}
\newenvironment{repeat-lemma}[1]{%
  \renewcommand\themanuallemmainner{#1}%
  \manuallemmainner{}
}{\endmanuallemmainner}

\newcommand{\SUBSECTION}[1]{%
  \ifdesperateforspace \bigskip \noindent {\bf #1.}
  \else\subsection{#1}
  \fi}
\newcommand{\SUBSUBSECTION}[1]{%
  \ifdesperateforspace \bigskip \noindent {\it #1.}
  \else\subsection{#1}
  \fi}

\ifdraft{}
\usepackage[notion,xcolor]{knowledge}
\else{}
\usepackage[notion,xcolor,electronic]{knowledge}
\fi{}

\knowledgestyle{notion}{color=knowledgelink}

\knowledge{notion}
  | FGen
  | free generator

\knowledge{notion}
  | simplified form

\knowledge{notion}
  | derivative of a free generator
  | free generator derivative
  | generator derivative
  | derivative@GEN
  | derivatives@GEN

\knowledge{notion}
  | nullability for free generators
  | nullability@GEN

\knowledge{notion}
  | language of a free generator
  | language interpretation

\knowledge{notion}
  | parser interpretation
  | parser interpretation of a free generator

\knowledge{notion}
  | choice distribution
  | choice distribution of a free generator

\knowledge{notion}
  | generator interpretation

\knowledge{notion}
  | generator equivalence

\knowledge{notion}
  | cgs

\knowledge{url={https://hackage.haskell.org/package/base-4.15.0.0/docs/Control-Applicative.html}}
  | applicative

\knowledge{url={https://hackage.haskell.org/package/base-4.16.0.0/docs/Data-Functor.html}}
  | functor

\newrobustcmd{\deriv}{\delta}
\newrobustcmd{\genderiv}{\kl[generator derivative]{\delta}}
\newrobustcmd{\nullable}{\nu}
\newrobustcmd{\gennullable}{\kl(GEN)[nullability]{\nu}}
\newrobustcmd{\lang}[1]{\kl[language of a free generator]{\mathcal{L}}\llbracket{} #1 \rrbracket{}}
\newrobustcmd{\parser}[1]{\kl[parser interpretation]{\mathcal{P}}\llbracket{} #1 \rrbracket{}}
\newrobustcmd{\choiceDist}[1]{\kl[choice distribution]{\mathcal{C}}\llbracket{} #1 \rrbracket{}}
\newrobustcmd{\valueDist}[1]{\kl[generator interpretation]{\mathcal{G}}\llbracket{} #1 \rrbracket{}}
\newrobustcmd{\modValueDist}[1]{\kl[generator interpretation]{\overline{\mathcal{G}}}\llbracket{} #1 \rrbracket{}}
\newrobustcmd{\fullalgo}{{\sc Choice Gradient Sampling}}
\newrobustcmd{\shortalgo}{{\sc CGS}}

\DeclareMathOperator{\fmap}{\kl[functor]{\langle{} \${} \rangle{}}}
\DeclareMathOperator{\ap}{\kl[applicative]{\langle{}\!\!\ast{}\!\!\rangle{}}}
\DeclareMathOperator{\gequiv}{\kl[generator equivalence]{\equiv}}

\usepackage{listings}
\usepackage{etoolbox}
\usepackage{graphicx}
\usepackage{enumitem}
\usepackage[group-separator={,}, group-minimum-digits={3}]{siunitx}

\usepackage[noend]{algpseudocode}
\algblockdefx[Hask]{Hask}{EndHask}%
[2]{#1~#2 =}%
{}
\algblockdefx[Hask]{HaskApprox}{EndHask}%
[3]{#1(#2~#3) $\approx$}%
{}

\begin{document}

\title[Parsing Randomness]{Parsing Randomness}
\subtitle{Unifying and Differentiating Parsers and
  Random Generators}
\titlenote{This paper was originally submitted to PLDI'22 and not accepted. The
  reviewers felt that the ideas were good, but that the presentation was
  incomplete.}


\author{Harrison Goldstein}
\orcid{0000-0001-9631-1169}             
\affiliation{
  \institution{University of Pennsylvania}            
  \country{Philadelphia, PA, USA}                    
}
\email{hgo@seas.upenn.edu}          

\author{Benjamin C. Pierce}
\orcid{0000-0001-7839-1636}             
\affiliation{
  \institution{University of Pennsylvania}           
  \country{Philadelphia, PA, USA}                   
}
\email{bcpierce@cis.upenn.edu}         

\begin{abstract}
  ``A generator is a parser of randomness.'' This
  perspective on generators for random data structures is
  well established as folklore
  in the programming languages community, but
  it has apparently never been formalized, nor
  have its consequences been deeply explored.

  We present {\em free generators}, which unify
  parsing and generation using a common structure that
  makes the relationship between the two concepts
  precise. Free generators lead naturally to a
  proof that a large class of generators can be factored into a
  parser plus a distribution over choice sequences.
  Further, free generators support a notion of
  {\em derivative}, analogous to familiar Brzozowski derivatives of
    formal languages, that allows analysis tools to
  ``preview'' the effect
  of a particular generator choice. This, in turn,
  gives rise to a novel algorithm for generating data structures
  satisfying user-specified preconditions.
\end{abstract}


\maketitle

\section{Introduction}\label{sec:intro}
``A generator is a parser of randomness\dots''  It's one of those observations
that's totally puzzling right up to the moment it becomes totally obvious:
a random generator---such as might be found in
a property-based
testing tool like QuickCheck~\cite{DBLP:conf/icfp/ClaessenH00}---is a
transformer from a series of random
choices into a data
structure, just as a parser transforms a series of
characters into a data structure.

Although this connection may be obvious once it is pointed out, few actually
think of generators this way. Indeed, to our
knowledge the framing of random generators as
parsers has never been explored
formally.
But this is a shame! The relationship between these fundamental
concepts deserves a deeper look.


A generator is a program that builds a data
structure by making a sequence of random
choices---those choices are the key. A ``traditional''
generator makes decisions using a stored source of
randomness (e.g., a seed) that it consults and updates whenever it must make a
choice.  Equivalently, if we like, we can pre-compute a list of choices
and pass it in to the generator, which gradually walks down the list whenever
it needs to make random decisions.
In this mode of operation, the
generator is effectively  {\em parsing\/} the sequence of choices
into a data structure!

%

To connect generators and parsers,
we introduce a data structure called a {\em free
generator\/} that can be interpreted as {\em either\/} a
generator or as a parser.  Free generators have a rich theory; in particular, we
can use them to prove
that a subset of generator programs can be
factored into a parser and a distribution over
sequences of choices.

Besides clarifying folklore, free generators admit
transformations that cannot be implemented for
standard generators and parsers. A particularly exciting one
is a notion of {\em derivative\/}
which modifies a generator by asking the question:
``what would this generator look like after it
makes choice $c$?'' The derivative gives a way of
previewing a particular choice to determine how
likely it is to lead us to useful values.

We use derivatives of free generators to
tackle a well-known
problem---we call it the {\em valid
generation problem}. The challenge is to generate
a large number of random values that satisfy some
validity condition. This problem comes up often in
property-based testing, where the validity
condition is the precondition of some functional
specification. Since generator derivatives give a
way of previewing the effects of a particular
choice, we can use {\em gradients\/} (derivatives
with respect to a vector of choices) to preview all possible
choices and pick a promising one. This leads us to an elegant algorithm for
turning a na\"\i ve free generator into one that only generates valid values.

In \S\ref{sec:high-level} below, we introduce the ideas behind free generators and
the operations that can be defined on them.  We then present
our main contributions:
\begin{itemize}
\item We formalize the folklore analogy between parsers and generators
  using {\em free generators}, a novel class of structures that make choices
  explicit and support syntactic transformations (\S\ref{sec:free}).
  We use free generators to prove that every
    ``applicative'' generator can factored into
  a parser and a probability distribution.
\item We exploit free generators to to transport
  an idea from formal languages---the {\em
    Brzozowski derivative\/}---to the context of
  generators (\S\ref{sec:derivatives}).
\item To illustrate the potential applications of these formal results,
  we present an algorithm that uses
  derivatives to turn a na\"\i ve generator into
  one that produces only values satisfying a
  Boolean precondition (\S\ref{sec:grad}). Our
  algorithm performs well on simple
  benchmarks, in most cases producing more than
  twice as many valid values as a na\"\i ve
  ``rejection sampling'' generator in the same amount of time
  (\S\ref{sec:evaluation}).
\end{itemize}
We conclude with related and future work
(\S\ref{sec:related} and \S\ref{sec:future}).

\section{The High-Level Story}\label{sec:high-level}
Let's take a walk in the forest before we dissect the trees.

\SUBSECTION{Generators and Parsers}
Consider the generator \lstinline{genTree} in
Figure~\ref{fig:genTree-parseTree}, which
produces random binary trees of Booleans like
\begin{center}
  \lstinline{Node True Leaf Leaf} \quad and \\
  \lstinline{Node True Leaf (Node False Leaf Leaf)},
\end{center}
up to a given height $h$, guided by a series of
random coin {\sf flip}s.
\footnote{Program synthesis
  experts might wonder
why we represent generators as programs of this
form, rather than, for example, PCFGs. Our work
may very well translate to grammar-based generators, but we
chose to target ``applicative'' generator programs because
they are more expressive and more familiar for
{\sc QuickCheck}-style testing.}

\begin{figure}[h]
  \begin{algorithmic}
    \Hask{$\textsf{genTree}$}{$h$}
      \If{$h = 0$}
        \State{$\textbf{return}~\textsf{Leaf}$}
      \Else{}
        \State{$c \gets \textsf{flip}()$}
        \If{$c == \textsf{Heads}$} $\textbf{return}~\textsf{Leaf}$ \EndIf{}
        \If{$c == \textsf{Tails}$}
          \State{$c \gets \textsf{flip}()$}
          \If{$c == \textsf{Heads}$} $x \gets \textsf{True}$ \EndIf{}
          \If{$c == \textsf{Tails}$} $x \gets \textsf{False}$ \EndIf{}
          \State{$l \gets \textsf{genTree}~(h - 1)$}
          \State{$r \gets \textsf{genTree}~(h - 1)$}
          \State{$\textbf{return}~\textsf{Node}~x~l~r$}
        \EndIf{}
      \EndIf{}
    \EndHask{}
  \end{algorithmic}
  \begin{algorithmic}
    \Hask{$\textsf{parseTree}$}{$h$}
    \If{$h = 0$}
      \State{$\textbf{return}~\textsf{Leaf}$}
    \Else{}
      \State{$c \gets \textsf{consume}()$}
        \If{$c == \cc{l}$} $\textbf{return}~\textsf{Leaf}$ \EndIf{}
        \If{$c == \cc{n}$}
          \State{$c \gets \textsf{consume}()$}
          \If{$c == \cc{t}$} $x \gets \textsf{True}$ \EndIf{}
          \If{$c == \cc{f}$} $x \gets \textsf{False}$
          \Else{} $\textbf{fail}$
          \EndIf{}
          \State{$l \gets \textsf{parseTree}~(h - 1)$}
          \State{$r \gets \textsf{parseTree}~(h - 1)$}
          \State{$\textbf{return}~\textsf{Node}~x~l~r$}
        \Else{} $\textbf{fail}$
      \EndIf{}
    \EndIf{}
    \EndHask{}
  \end{algorithmic}
  \vspace{-4mm}
  \caption{A generator and a parser for Boolean binary trees. \bcp{We can save
      four lines here by removing carriage returns after then and else.  Also
      figure 2.}}\label{fig:genTree-parseTree}
\end{figure}

Now, consider \lstinline{parseTree} (also in
Figure~\ref{fig:genTree-parseTree}), which parses
a string over the characters
\cc{n}, \cc{l}, \cc{t}, and \cc{f} into a tree. The parser
turns
\begin{center}
  \cc{ntll} into \lstinline{Node True Leaf Leaf} \quad and \\
  \cc{ntlnfll} into \lstinline{Node True Leaf (Node False Leaf Leaf)}.
\end{center}
It consumes the input string character by character with {\sf
  consume} and uses the characters to decide what
to do next.


Obviously, there is considerable structural similarity between
\lstinline{genTree} and \lstinline{parseTree}.
One apparent difference lies in the way they make choices and the ``labels'' for
those choices:
in \lstinline{genTree}, choices are made randomly
during the execution of the program and are marked
by sides of a coin, while
in \lstinline{parseTree} the choices are made ahead
of time and manifest as the characters in the
input string. But this difference is rather superficial.

\SUBSECTION{Free Generators}
We can unify random generation with parsing by
abstracting both into a single data
structure. For this, we introduce \kl[free
generator]{free generators}.\footnote{This
document uses the {\tt knowledge} package in
\LaTeX{} to make definitions interactive. Readers
viewing the PDF electronically can click on
technical terms and symbols to see where they are
defined in the document.} Free generators are syntactic
structures (a bit like abstract syntax trees) that can be {\em interpreted\/} as
programs that either generate or parse.
Observe the structural similarities between {\sf fgenTree} and the programs in
Figure~\ref{fig:genTree-parseTree}.

\begin{figure}[h]
  \begin{algorithmic}
    \Hask{$\textsf{fgenTree}$}{$h$}
    \If{$h = 0$}
    \State{\lstinline{Pure Leaf}}
    \Else{}
\begin{lstlisting}
       Select
         [ ($\cc{l}$, Pure Leaf),
           ($\cc{n}$, MapR
             (Pair (Select
                     [ ($\cc{t}$, Pure True),
                       ($\cc{f}$, Pure False) ])
                   (Pair (fgenTree ($h$ - 1))
                         (fgenTree ($h$ - 1))))
             (\ (x, (l, r)) -> Node x l r)) ]
\end{lstlisting}
    \EndIf{}
    \EndHask{}
  \end{algorithmic}
  \vspace{-4mm}
  \caption{A free generator for binary trees of Booleans.\bcp{We can save a
      couple lines here by removing carriage returns after then and else.  Also,
  let's make sure this winds up on the same page as figure
  1.}}\label{fig:fgenTree}
\end{figure}

While the free generator \lstinline{fgenTree $h$}
is just a data structure, its shape is much the same
as {\sf genTree} and {\sf parseTree}.
A {\sf Pure} node in the free generator
corresponds roughly to \lstinline{return}; it
represents a pure value that makes no choices.
{\sf MapR} takes two
arguments, a free generator and a function that
will eventually be applied to the result of
generation / parsing. The {\sf Pair}
constructor maps to sequencing the original programs: it
generates / parses using its first argument, then
does the same with its second argument, and
finally pairs the results together. Finally---the real magic---lies in how we
interpret the
\lstinline{Select} structure. When we want a
generator, we treat it as making a uniform random
choice, and when we want a parser we treat it as
consuming a character $c$ and checking it
against the first elements of the pairs.

In \S\ref{sec:free} we
give formal definitions of free generators, along
with several interpretation functions. We write
$\valueDist{\cdot}$ for the \kl{generator
interpretation} of a free generator and
$\parser{\cdot}$ for the \kl{parser
interpretation}. In other
words,
\begin{center}
  $\valueDist{\textsf{fgenTree}~5}~\approx~\textsf{genTree}~5$ \quad and
  $\parser{\textsf{fgenTree}~5}~\approx~\textsf{parseTree}~5$.
\end{center}

Now let's consider how the generator and parser interpretations
relate. The key lies in one final interpretation function,
$\choiceDist{\cdot}$, which yields the \kl{choice distribution}.
Intuitively, the choice distribution interpretation produces the
set of
sequences of choices that the generator
interpretation can make, or equivalently the set of
sequences that the parser interpretation can
parse.

The choice distribution interpretation is used below in
Theorem~\ref{thm:coherence} to connect
parsing and generation. The theorem says
that for any free generator $g$,
\[ \parser{g} \fmap \choiceDist{g} \approx \valueDist{g} \]
where $\fmap$ is a ``mapping''
operation that applies a function to samples from
a distribution. Since many normal {\sc
QuickCheck} generators can also be written as free
generators, another way to read this theorem is
that such generators can be factored into two
pieces: a distribution over choice sequences
(given by $\choiceDist{\cdot}$), and a parser of
those sequences (given by $\parser{\cdot}$). This
precisely formalizes the intuition that ``A generator is a parser of
randomness.''

\SUBSECTION{Derivatives of Free Generators}
But wait, there's more!
Since a free
generator defines a parser, it also defines a formal
language: we write $\lang{\cdot}$ for this \kl{language
interpretation} of a free generator. The language of a free generator
is the set of choice sequences that it can parse
(or make).

Viewing free generators this way
suggests some interesting ways that free
generators might be manipulated.
In particular, formal languages come with a notion of {\em
derivative}, due to
Brzozowski~\cite{brzozowski1964derivatives}. Given a
language $L$, the Brzozowski derivative of $L$ is
\[
\deriv_{c}L = \{ s \mid c \cdot s \in L \} .
\]
That is, the derivative of $L$ with respect
to $c$ is all the strings in $L$ that start with $c$,
with the first $c$ removed.

%

Conceptually, the
derivative of a parser with respect to a character
$c$ is whatever parser remains after $c$ has
just been parsed.
For example, the derivative of
\lstinline{parseTree 5} with respect to \cc{n} is:
\bcp{delete the carriage return after the approx
  (also below)...}
\\\\
\begin{minipage}[c]{.95\textwidth}
  \begin{algorithmic}
    \HaskApprox{$\deriv_{\cc{n}}$}{$\textsf{parseTree}$}{5}
    \State{$c \gets \textsf{consume}()$}
    \If{$c == \cc{t}$} $x \gets \textsf{True}$ \EndIf{}
    \If{$c == \cc{f}$} $x \gets \textsf{False}$
    \Else{} $\textbf{fail}$
    \EndIf{}
    \State{$l \gets \textsf{parseTree}~4$}
    \State{$r \gets \textsf{parseTree}~4$}
    \State{$\textbf{return}~\textsf{Node}~x~l~r$}
    \EndHask{}
  \end{algorithmic}
\end{minipage}
\noindent After parsing the character
\cc{n}, the next step in the original parser is to parse either \cc{t} or
\cc{f} and then construct a {\sf Node}; the
derivative does just that.

Next let's take a derivative of the new parser
$\deriv_{\cc{n}}(\textsf{parseTree}~5)$---this time
with respect to \cc{t}: \\\\
\begin{minipage}[c]{.95\textwidth}
  \begin{algorithmic}
    \HaskApprox{$\deriv_{\cc{t}}\deriv_{\cc{n}}$}{$\textsf{parseTree}$}{5}
    \State{$l \gets \textsf{parseTree}~4$}
    \State{$r \gets \textsf{parseTree}~4$}
    \State{$\textbf{return}~\textsf{Node}~\textsf{True}~l~r$}
    \EndHask{}
  \end{algorithmic}
\end{minipage}
\noindent Now we have fixed the value {\sf True}
for $x$, and we can continue by making the
recursive calls to \lstinline{parseTree 4} and constructing the final tree.

Free generators have a closely related
notion of \kl(GEN){derivative}.
The
derivatives of the free
generator produced by {\sf fgenTree} look
almost identical to the ones that we saw above for {\sf
parseTree}: \\\\
\begin{minipage}[c]{.95\textwidth}
  \begin{algorithmic}
    \HaskApprox{$\genderiv_{\cc{n}}$}{$\textsf{fgenTree}$}{5}
\begin{lstlisting}
    MapR
      (Pair (Select
              [ ($\cc{t}$, Pure True),
                ($\cc{f}$, Pure False) ])
            (Pair (fgenTree 4)
                  (fgenTree 4)))
      (\ (x, (l, r)) -> Node x l r)
\end{lstlisting}
    \EndHask{}
  \end{algorithmic}
\end{minipage}

\begin{minipage}[c]{.95\textwidth}
  \begin{algorithmic}
    \HaskApprox{$\genderiv_{\cc{t}}\genderiv_{\cc{n}}$}{$\textsf{fgenTree}$}{5}
\begin{lstlisting}
    MapR
      (Pair (fgenTree 4)
            (fgenTree 4))
      (\ (l, r) -> Node True l r)
\end{lstlisting}
    \EndHask{}
  \end{algorithmic}
\end{minipage}

Moreover, like derivatives of regular
expressions and context-free grammars, derivatives of free generators can be
computed by a simple syntactic transformation.
In \S\ref{sec:derivatives} we define a
procedure for computing the derivative of a free
generator and prove it correct, in the sense that,
for all free generators $g$,
\[ \deriv_{c} \lang{g} = \lang{\genderiv_{c}g}. \]
In other words, the derivative of the language of
$g$ is equal to the language of the derivative of
$g$. (See Theorem~\ref{thm:deriv-consistent}.)

\SUBSECTION{Putting Free Generators to Work}
The derivative of a free generator is intuitively
{\em the generator that remains after a
particular choice}. This gives us a way of ``previewing''
the effect of making a choice by looking at the
generator after fixing that choice.

In \S\ref{sec:grad} and \S\ref{sec:evaluation} we
present and evaluate an algorithm called
\kl[cgs]{\fullalgo{}} that uses free generators to
address the {\em valid generation problem}. Given a
validity predicate on a data structure, the goal
is to generate as many unique, valid structures as
possible in a given amount of time.
Given a simple free generator, our algorithm uses
derivatives to evaluate choices and search for
valid values.

We evaluate our algorithm on four
small benchmarks, all standard in the
property-based testing literature. We compare our
algorithm to rejection sampling---sampling from a
na\"\i ve generator and discarding
invalid results---as a simple but useful
baseline for understanding how well or algorithm
performs. Our algorithm does remarkably well on
all but one benchmark, generating more than twice
as many valid values as rejection sampling in the
same period of time.

\section{Free Generators}\label{sec:free}
We now turn to developing the theory of {free generators}, beginning with some
background on applicative abstractions for
parsing and random generation.

\SUBSECTION{Background: Applicative Parsers and Generators}
In \S\ref{sec:high-level} we represented generators
and parsers with pseudo-code. Here we flesh out the details.
We present all definitions as {\sc Haskell} programs, both for
the sake of
concreteness and also because {\sc Haskell}'s
abstraction features (e.g., typeclasses) allow us to focus on the
key concepts.  {\sc Haskell} is a lazy functional language, but our
results are also applicable to eager functional
languages and imperative languages.

We represent both generators and parsers using
{\em applicative
  functors\/}~\cite{mcbride2008applicative}\footnote{For Haskell experts: we
  choose to focus on applicatives, not monads,
  to simplify our development and avoid some
  efficiency issues in \S\ref{sec:derivatives} and
  \S\ref{sec:grad}. Much of what we present
  should generalize to monadic generators as well.}
At a
high level, an applicative functor
is a type constructor \lstinline{f} with
operations:
\begin{lstlisting}
  (fmap) :: (a -> b) -> f a -> f b
  pure :: a -> f a
  (<*>) :: f (a -> b) -> f a -> f b
\end{lstlisting}
When it might not be clear which applicative
functor we mean, we prefix the operator with the
name of the functor (e.g., \lstinline{Gen.fmap}).
These operations are mainly useful as a way to
apply functions to values inside of some data
structure or computation.
For example, the
idiom ``\lstinline{g fmap x <*> y <*> z}'' applies
a pure function \lstinline{g} to the values in
three structures \lstinline{x}, \lstinline{y}, and
\lstinline{z}.

We can use these operations to define {\sf
genTree} like we would in {\sc
QuickCheck}~\cite{DBLP:conf/icfp/ClaessenH00},
since the {\sc
QuickCheck} type constructor {\sf Gen}, which represents
generators, is an applicative functor:\\
\begin{minipage}[c]{.95\textwidth}
\begin{lstlisting}
  genTree :: Int -> Gen Tree
  genTree 0 = pure Leaf
  genTree $h$ =
    oneof [ pure Leaf,
            Node fmap genInt
                 <*> genTree ($h$ - 1)
                 <*> genTree ($h$ - 1) ]
\end{lstlisting}
\end{minipage}
Here, {\sf pure} is the trivial
generator that always generates the same value,
and \lstinline{Node fmap g1 <*> g2 <*> g3} means
apply the constructor {\sf Node} to three
sub-generators to produce a new generator.
Operationally, this means sampling {\sf x1} from
{\sf g1}, {\sf x2} from {\sf g2}, and {\sf x3}
from {\sf g3}, and then constructing
\lstinline{Node x1 x2 x3}.  
Notice that we need one extra function beyond the
applicative interface: {\sf oneof} makes a uniform
choice between generators, just as we saw in the
pseudo-code.

We can do the same thing for {\sf parseTree},
using combinators inspired by libraries like {\sc
Parsec}~\cite{leijen2001parsec}:\\
\begin{minipage}[c]{.95\textwidth}
\begin{lstlisting}
  parseTree :: Int -> Parser Tree
  parseTree 0 = pure Leaf
  parseTree $h$ =
    choice [ ($\cc{l}$, pure Leaf),
             ($\cc{n}$, Node fmap parseInt
                      <*> parseTree ($h$ - 1)
                      <*> parseTree ($h$ - 1)) ]
\end{lstlisting}
\end{minipage}
In this context, \lstinline{pure} is a parser that
consumes no characters and never fails. It just
produces the value passed to it. We can interpret
\lstinline{Node fmap p1 <*> p2 <*> p3} as running
each sub-parser in sequence (failing if any of
them fail) and then wrapping the results in the
{\sf Node} constructor. Finally, we have replaced
{\sf oneof} with {\sf choice}, but the idea is the
same: choose between sub-parsers.

Parsers like this have type
\lstinline{String -> Maybe (a, String)}.
They can be applied to a string to obtain either
\lstinline{Nothing} or \lstinline{Just (a, s)},
where \lstinline{a} is the parse result and
\lstinline{s} contains any extra characters.

\SUBSECTION{Representing Free Generators}
\AP{} With the applicative interface in mind, we
can now give the formal definition of a
\intro{free generator}.\footnote{For algebraists:
free generators are ``free,'' in the sense that
they admit unique structure-preserving maps to
other ``generator-like'' structures. In
particular, the $\valueDist{\cdot}$ and
$\parser{\cdot}$ maps are canonical. For the sake
of space, we do not explore these ideas
further here.}

\SUBSUBSECTION{Type Definition}
We represent free generators as an inductive data
type, {\sf FGen}, defined as:\\
\begin{minipage}[c]{.95\textwidth}
\begin{lstlisting}
  data FGen a where
    Void :: FGen a
    Pure :: a -> FGen a
    Pair :: FGen a -> FGen b -> FGen (a, b)
    Map :: (a -> b) -> FGen a -> FGen b
    Select :: List (Char, FGen a) -> FGen a
\end{lstlisting}
\end{minipage}
These constructors form an abstract syntax tree
with nodes that roughly correspond to the functions
in the applicative interface. Clearly {\sf Pure}
represents {\sf pure}. {\sf Pair} is a slightly
different form of \lstinline{<*>}; one is
definable from the other, but this version makes
more sense as a data constructor. {\sf Map}
corresponds to \lstinline{fmap} (but note
  that the arguments to {\sf Map} are flipped relative
  to {\sf MapR} from \S\ref{sec:high-level}).
Finally, {\sf Select}
subsumes both {\sf oneof} and {\sf
choice}: it might mean either, depending on the
interpretation. Finally {\sf Void} represents an
always-failing parser or a generator of nothing.

Free generators draw inspiration from {\em free
applicative functors\/}~\cite{capriotti2014free}.
As with free
applicative functors, we can write transformations
\lstinline{FGen a -> f a} for any \lstinline{f}
with similar structure. This fact motivates the
rest of this section.

\SUBSUBSECTION{Language of a Free Generator}
\AP{} The \intro{language of a free
generator} is the set of choice sequences that it
can make or parse. It is defined
recursively, by cases:\\
\begin{minipage}[c]{.95\textwidth}
\begin{lstlisting}
  cL[[$\cdot$]] :: FGen a -> Set String
  cL[[Void]]      = $\varnothing$
  cL[[Pure a]]    = $\varepsilon$
  cL[[Map f x]]   = cL[[x]]
  cL[[Pair x y]]  = {s $\cdot$ t | s $\in$ cL[[x]] $\wedge$ t $\in$ cL[[y]]}
  cL[[Select xs]] = {c $\cdot$ s | (c, x) $\in$ xs $\wedge$ s $\in$ cL[[x]]}
\end{lstlisting}
\end{minipage}

\SUBSUBSECTION{Smart Constructors and Simplified Forms}
\AP{} Free generators admit a useful
\intro{simplified form}. To ensure that generators
are simplified, we can require that free generators
be built using {\em smart constructors}.

In particular, instead of using \lstinline{Pair} directly, we can pair
free generators with the smart constructor $\otimes$\\
\begin{minipage}[c]{.95\textwidth}
\begin{lstlisting}
  ($\otimes$) :: FGen a -> FGen b -> FGen (a, b)
  Void   $\otimes$ _      = Void
  _      $\otimes$ Void   = Void
  Pure a $\otimes$ y      = (\b -> (a, b)) fmap y
  x      $\otimes$ Pure b = (\a -> (a, b)) fmap x
  x      $\otimes$ y      = Pair x y
\end{lstlisting}
\end{minipage}
which makes sure that \lstinline{Void} and
\lstinline{Pure} are collapsed with respect to
\lstinline{Pair}. For
example, \lstinline{Pure a $\otimes$ Pure b}
collapses to \lstinline{Pure (a, b)}.

The smart constructor \lstinline{fmap} is a version of
\lstinline{Map} that does similar
collapsing:\\
\begin{minipage}[c]{.95\textwidth}
\begin{lstlisting}
  (fmap) :: (a -> b) -> FGen a -> FGen b
  f fmap Void   = Void
  f fmap Pure a = Pure (f a)
  f fmap x      = Map f x
\end{lstlisting}
\end{minipage}

We define \lstinline{pure} and \lstinline{<*>} so as to
make \lstinline{FGen} an applicative functor:\\
\begin{minipage}[c]{.95\textwidth}
\begin{lstlisting}
  pure :: a -> FGen a
  pure = Pure
  (<*>) :: FGen (a -> b) -> FGen a -> FGen b
  f <*> x = (\ (f, x) -> f x) fmap (f $\otimes$ x)
\end{lstlisting}
\end{minipage}

The smart constructor {\sf Select} looks like
this:\\
\begin{minipage}[c]{.95\textwidth}
\begin{lstlisting}
  select :: List (Char, FGen a) -> FGen a
  select xs =
    case filter (\ (_, p) -> p $\neq$ Void) xs of
      xs | xs == [] || hasDups (map fst xs) -> $\bot$
      xs -> Select xs
\end{lstlisting}
\end{minipage}
This smart constructor filters out any
sub-generators that are {\sf Void} (since those
are functionally useless), and it fails (returning
$\bot$) if the final list of sub-generators is
empty or if it has duplicated choice tags. This
ensures that the operations on generators defined
later in this section will be well formed.

Finally, we define a smart constructor \lstinline{void = Void} for consistency.

When a generator is built using a finite tree of smart constructors, we say it
is in simplified form.
\footnote{
  In strict languages, the finiteness requirement
  for simplified forms is needed because it guarantees that
  the program producing the free generator will
  terminate. In lazy languages, one can write
  infinite co-inductive data structures; nevertheless, we focus
  on finite free generators, because, in practice,
  one rarely wants to generate values of arbitrary
  size.
}

\SUBSUBSECTION{Examples}\label{subsec:fg-examples}
We saw a version of {\sf fgenTree} in
\S\ref{sec:high-level} that was written out
explicitly as an AST\@. Here's how it would
actually be done in our framework, with smart constructors:\\
\begin{minipage}[c]{.95\textwidth}
\begin{lstlisting}
  fgenTree :: Int -> FGen Tree
  fgenTree 0 = pure Leaf
  fgenTree $h$ =
    select [ ($\cc{l}$, pure Leaf),
             ($\cc{n}$, Node fmap fgenInt
                     <*> fgenTree ($h$ - 1)
                     <*> fgenTree ($h$ - 1)) ]
\end{lstlisting}
\end{minipage}
Recall that {\sf fgenTree} is meant to subsume both
{\sf genTree} and {\sf parseTree}. The height
parameter $h$ is used to cut
off the depth of trees and prevent the resulting free
generators from being infinitely deep.


Here is another example of a free generator that produces
random terms of the simply-typed lambda-calculus:\\
\begin{minipage}[c]{.95\textwidth}
\begin{lstlisting}
  fgenExpr :: Int -> FGen Expr
  fgenExpr 0 =
    select [ ($\cc{i}$, Lit fmap fgenInt),
             ($\cc{v}$, Var fmap fgenVar) ]
  fgenExpr h =
    select [ ($\cc{i}$, Lit fmap fgenInt),
             ($\cc{p}$, Plus fmap fgenExpr (h - 1)
                     <*> fgenExpr (h - 1)),
             ($\cc{l}$, Lam fmap fgenType
                    <*> fgenExpr (h - 1)),
             ($\cc{a}$, App fmap fgenExpr (h - 1)
                    <*> fgenExpr (h - 1)),
             ($\cc{v}$, Var fmap fgenVar) ]
\end{lstlisting}
\end{minipage}
Structurally this is quite similar to the previous
generator; it just has more cases and more
choices. This lambda calculus uses de Bruijn
indices for variables and has integers and
functions as values. This is a useful example
because while syntactically valid terms in this
language are easy to generate (as we just did), it is more
difficult to generate only well-typed terms. We
use this example as one of our case studies in
\S\ref{sec:evaluation}.

\SUBSECTION{Interpreting Free Generators}
A free generator does not {do} anything on
its own---it is simply a data structure. We next define the
interpretation functions that we mentioned in
\S\ref{sec:high-level} and prove a theorem
linking those interpretations together.

\SUBSUBSECTION{Free Generators as Generators of Values}
\AP{} The first and most natural way to interpret
a free generator is as a {\sc QuickCheck}
generator---that is, as a distribution over data
structures. We define the \intro{generator
interpretation} of a free generator to be:\\
\begin{minipage}[c]{.95\textwidth}
\begin{lstlisting}
  cG[[$\cdot$]] :: FGen a -> Gen a
  cG[[Void]]      = $\bot$
  cG[[Pure v]]    = Gen.pure v
  cG[[Map f x]]   = f Gen.fmap cG[[x]]
  cG[[Pair x y]]  =
    (\x y -> (x, y)) Gen.fmap cG[[x]] Gen.<*> cG[[y]]
  cG[[Select xs]] =
    oneof (map (\ (_, x) -> cG[[x]]) xs)
\end{lstlisting}
\end{minipage}
Note that the operations on the right-hand side of this
definition are {\em not\/} free generator
constructors; they are {\sc QuickCheck} generator
operations. This definition maps ``AST nodes'' to
the equivalent interpretation implemented by {\sf
Gen}. In the case for {\sf Pair}, notice the pattern
that we described earlier in this section. The
code
\begin{lstlisting}
(\x y -> (x, y)) Gen.fmap cG[[x]] Gen.<*> cG[[y]]
\end{lstlisting}
pairs the results of \lstinline{cG[[x]]} and
\lstinline{cG[[y]]} in a tuple via applicative
idiom ``\lstinline{g fmap x <*> y}''.

One detail worth noting is that the interpretation
behaves poorly (it diverges) on \textsf{Void}; fortunately, the following lemma
shows that this does not cause problems
in practice:
\begin{lemma}\label{lem:norm-not-void}
  If a \kl{free generator} $g$ is \kl[simplified form]{simplified}, then
  \[\text{$g$ contains \textsf{Void}} \iff g = \textsf{Void}.\]
\end{lemma}
\begin{proof}
  By induction on the structure of $g$ and inspection of the smart constructors.
\end{proof}
\noindent Thus we can conclude that, as long as $g$
is in simplified form and not \textsf{Void},
$\valueDist{g}$ is defined.
\begin{example}
  $\valueDist{\textsf{fgenTree}~5}$ is equivalent to \lstinline{genTree 5}.
\end{example}

\SUBSUBSECTION{Free Generators as Parsers of Choice Sequences}
\AP{} Now we come to the main technical point of the paper.
We can make use of the character labels in the {\sf Select} nodes
using a free
generator's \intro{parser interpretation}---in
other words, we can view a free generator as a parser of choices. The
translation looks like this:\\
\begin{minipage}[c]{.95\textwidth}
\begin{lstlisting}
  cP[[$\cdot$]] :: FGen a -> Parser a
  cP[[Void]]      = \s -> Nothing
  cP[[Pure a]]    = Parser.pure a
  cP[[Map f x]]   = f Parser.fmap cP[[x]]
  cP[[Pair x y]]  =
    (\x y -> (x, y)) Parser.fmap cP[[x]] Parser.<*> cP[[y]]
  cP[[Select xs]] =
    choice (map (\ (c, x) -> (c, cP[[x]])) xs)
\end{lstlisting}
\end{minipage}
\noindent This definition uses the representation
of parsers as functions of type
\lstinline{String -> Maybe (a, String)}
that we saw earlier.
\begin{example}
  $\parser{\textsf{fgenTree}~5}$ is
  equivalent to \lstinline{parseTree 5}.
\end{example}

\SUBSUBSECTION{Free Generators as Generators of Choice Sequences}
\AP{} Our final interpretation of free generators captures the part of the
generator ``missed'' by the parser---it represents the distribution with which
the generator makes choices.
We define the \intro{choice distribution} of a free generator to be:\\
\begin{minipage}[c]{.95\textwidth}
\begin{lstlisting}
  cC[[$\cdot$]] :: FGen a -> Gen String
  cC[[Void]]      = $\bot$
  cC[[Pure a]]    = Gen.pure $\varepsilon$
  cC[[Map f x]]   = cC[[x]]
  cC[[Pair x y]]  =
    (\s t -> s $\cdot$ t) Gen.fmap cC[[x]] Gen.<*> cC[[y]]
  cC[[Select xs]] =
    oneof (map (\ (c, x) -> (c $\cdot$) Gen.fmap cC[[x]]) xs)
\end{lstlisting}
\end{minipage}
We can
think of the result of this interpretation as a
distribution over $\lang{g}$. The language of a
free generator is exactly those choice sequences
that the generator interpretation can make and the
parser interpretation can parse.

\SUBSUBSECTION{Factoring Generators}
These different interpretations of free generators
are closely related to one another; in
particular, we can reconstruct $\valueDist{\cdot}$
from $\parser{\cdot}$ and $\choiceDist{\cdot}$. In
essence, this means that a free generator's
generator interpretation can be factored into a
distribution over choice sequences plus a parser of
those sequences.

\AP{} To make this more precise, we need a notion
of equality for generators like the ones produced
via $\valueDist{\cdot}$. We say two {\sc
QuickCheck} generators are \intro[generator
equivalence]{equivalent}, written $g_1 \gequiv
g_2$, if and only if the generators represent the
same distribution over values. This is coarser
notion than program equality, since
two generators might produce the same distribution
of values in different ways.

With this in mind, we can state and prove the
relationship between different interpretations of
free generators:
\begin{theorem}[Factoring]\label{thm:coherence}
  Every \kl[simplified form]{simplified} \kl{free
    generator} can be factored into a parser and a
  distribution over choice sequences. In other
  words, for all simplified free generators $g \neq
  \textsf{Void}$,
  \[
    \parser{g} \fmap \choiceDist{g} \gequiv (\lambda x \to \textsf{Just}~(x, \varepsilon)) \fmap
    \valueDist{g}.
  \]
\end{theorem}
\begin{proofsketch}
  By induction on the structure of $g$.
  \begin{itemize}
  \item[Case] $g = \textsf{Pure a}$. Straightforward.
  \item[Case] $g = \textsf{Map f x}$. Straightforward.
  \item[Case] $g = \textsf{Pair x y}$. This case is the most interesting one. The difficulty is that
    it is not immediately obvious why
    $\parser{\textsf{Pair x y}} \fmap \choiceDist{\textsf{Pair x y}}$
    should be a function of
    $\parser{\textsf{x}} \fmap \choiceDist{\textsf{x}}$
    and
    $\parser{\textsf{y}} \fmap \choiceDist{\textsf{y}}$.
    Showing the correct relationship requires a lemma that says that for any sequence $s$ generated
    by $\choiceDist{\textsf{x}}$ and an arbitrary sequence $t$, there is some $a$ such that
    $\parser{\textsf{x}}~(s \cdot t) = \textsf{Just}~(a, t)$.
  \item[Case] $g = \textsf{Select xs}$. The reasoning in this case is a bit subtle, since it
    requires certain operations to commute with \textsf{Select}, but the details are not
    particularly instructive.
  \end{itemize}
  See Appendix~\ref{appendix:coherence} for the full proof.
\end{proofsketch}

A natural corollary of Theorem~\ref{thm:coherence}
is the following:
\begin{corollary}
  Any finite applicative generator, $\gamma$,
  written in terms of pure functions, $\fmap$, {\sf
    pure}, $\ap$, and {\sf oneof}, can be factored
  into a parser and distribution over choice
  sequences.
\end{corollary}
\begin{proof}
  Translate $\gamma$ into a free generator,
  $g$, by replacing operations with the equivalent
  smart constructor. (For {\sf oneof}, draw unique
  labels for each choice and use {\sf select}.) By
  induction, $\valueDist{g} = \gamma$.

  The resulting free generator can be factored
  into a parser and a choice distribution via
  Theorem~\ref{thm:coherence}. Thus,
  \[
    (\lambda x \to \textsf{Just}~(x, \varepsilon))
    \fmap \gamma \gequiv \parser{g} \fmap \choiceDist{g},
  \]
  and $\gamma$ can be factored as desired.
\end{proof}
This corollary gives a concrete way to view the
connection between applicative generators and
parsers.

\SUBSECTION{Replacing a Generator's Distribution}\label{subsec:replacing-dist}
Since a generator $g$ can be factored using
$\choiceDist{\cdot}$ and $\parser{\cdot}$, we can
explore what it would look like to modify a
generator's distribution (i.e., change or replace
$\choiceDist{\cdot}$) without having to modify the
entire generator.

Suppose we have some other distribution that we
want our choices to follow. We can represent an external distribution as a
function from a history of choices to a generator
of next choices, together with a ``current'' history.
We write this type as:
\begin{lstlisting}
  type Dist = (String, String -> Gen (Maybe Char))
\end{lstlisting}
(If the choice function returns \lstinline{Nothing}, then generation stops.)

A \lstinline{Dist} may be arbitrarily complex:
it might contain information obtained
from example-based tuning, a machine learning
model, or some other automated tuning process. How
would we use such a distribution in place of the
standard distribution given by
$\choiceDist{\cdot}$?

The solution is to replace $\choiceDist{\cdot}$
with our new distribution to yield a modified
definition of the generator interpretation:\\
\begin{minipage}[c]{.95\textwidth}
\begin{lstlisting}
  cGhat[[$\cdot$]] :: (Dist, FGen a) -> Gen (Maybe a)
  cGhat[[((h, d), g)]] = cP[[g]] Gen.fmap genDist h
    where
      genDist h = d h >>= \x -> case x of
        Nothing -> Gen.pure h
        Just c -> genDist (h $\cdot$ c)
\end{lstlisting}
\end{minipage}
This definition exploits our new connection
between parsers and generators to obtain a new
generator interpretation via the parser
interpretation.

Since replacing a free generator's distribution
does not actually change the structure of the
generator, we can have a different distribution
for each use-case of the free generator. In a
property-based testing scenario, one could imagine
the tester fine-tuning a distribution for each
property, carefully optimized to find bugs
as quickly as possible.

\section{Derivatives of Free Generators}\label{sec:derivatives}
Next, we review the notion of
Brzozowski derivative in formal language theory
and show that a similar operation exists for
\kl[free generator]{free generators}. The way these derivatives fall out from
the structure of free
generators highlights the advantages of
taking the correspondence between generators and
parsers seriously.

\SUBSECTION{Background: Derivatives of Languages}
The {\em Brzozowski
derivative\/}~\cite{brzozowski1964derivatives} of a
formal language $L$ with respect to some choice
$c$ is defined as\bcp{$\deriv_{c} L$, or $\deriv_{c} (L)$?}
\[ \deriv_{c} L = \{s \mid c \cdot s \in L\}. \]
In other words, the derivative is the set of
strings in $L$ with $c$ removed from the front.
For example,
\[
  \deriv_{\cc{a}} \{\cc{abc}, \cc{aaa}, \cc{bba}\} = \{\cc{bc}, \cc{aa}\}.
\]

Many formalisms for defining languages support syntactic
transformations that correspond to Brzozowski
derivatives. For example, we can take the
derivative of a regular expression like this:
\begin{figure}[H]
  \centering
  \begin{subfigure}{0.59\columnwidth}
    \begin{align*}
      \deriv_{c} \varnothing &= \varnothing \\
      \deriv_{c} \varepsilon &= \varnothing \\
      \deriv_{c} \cc{c} &= \varepsilon \quad (c = \cc{c}) \\
      \deriv_{c} \cc{d} &= \varnothing \quad (c \neq \cc{d}) \\
      \deriv_{c} (r_1 + r_2) &= \deriv_{c} r_1 + \deriv_{c} r_2 \\
      \deriv_{c} (r_1 \cdot r_2) &= \deriv_{c} r_1 \cdot r_2 + \nullable r_1 \cdot \deriv_{c} r_2 \\
      \deriv_{c} (r^*) &= \deriv_{c}r \cdot r^*
    \end{align*}
  \end{subfigure}
  \begin{subfigure}{0.4\columnwidth}
    \begin{align*}
      \nullable \varnothing &= \varnothing \\
      \nullable \varepsilon &= \varepsilon \\
      \nullable \cc{c} &= \varnothing \\
      \nullable (r_1 + r_2) &= \nullable r_1 + \nullable r_2 \\
      \nullable (r_1 \cdot r_2) &= \nullable r_1 \cdot \nullable r_2 \\
      \nullable (r^*) &= \varepsilon
    \end{align*}
  \end{subfigure}
\end{figure}
\noindent \AP{} The $\nullable$ operator, used in
the ``$\cdot$'' rule and defined on the right, determines the
{\em nullability\/} of an expression (whether or
not it accepts $\varepsilon$). As one would hope,
if $r$ has language $L$, it is always the case
that $\deriv_{c} r$ has language $\deriv_{c} L$.

\SUBSECTION{The Free Generator Derivative}
\AP{} Since free generators define a language
(given by $\lang{\cdot}$),
can we take their derivatives? Yes, we can!

The
\intro[derivative of a free generator]{derivative}
of a free generator $g$ with respect to a character $c$, written
$\genderiv_{c}(g)$, is defined as follows:\\
\begin{minipage}[c]{.95\textwidth}
\begin{lstlisting}
  $\genderiv$ :: Char -> FGen a -> FGen a
  $\genderiv_{c}$Void        = void
  $\genderiv_{c}$(Pure v)    = void
  $\genderiv_{c}$(Map f x)   = f fmap $\genderiv_{c}$x
  $\genderiv_{c}$(Pair x y)  = $\genderiv_{c}$x $\otimes$ y
  $\genderiv_{c}$(Select xs) = if ($c$, x) $\in$ xs then x else void
\end{lstlisting}
\end{minipage}
Most of this definition should be intuitive. The
derivative of a generator that does not make a
choice (i.e., \lstinline{Void} and
\lstinline{Pure}) is \lstinline{void}, since the
corresponding language would be empty.
The derivative commutes with \lstinline{Map} since
the transformation affects choices, not the final
result. \lstinline{Select}'s derivative is
just the argument generator corresponding to the
appropriate choice.

The one potentially confusing case is the one for
\lstinline{Pair}. We have defined the derivative
of a pair of generators by taking the derivative
of the first generator in the pair and leaving the
second unchanged, which seems inconsistent with
the case for ``$\cdot$'' in the regular expression
derivative (what happens when the first
generator's language is nullable?). Luckily, our
\kl{simplified form} clears up the confusion: if
\lstinline{Pair x y} is in simplified form,
$\textsf{x}$ is not nullable. This is a simple
corollary of Lemma~\ref{lem:norm-null}.

\begin{lemma}\label{lem:norm-null}
  If a \kl{free generator} $g$ is in \kl{simplified form}, then either $g =
  \text{\lstinline{Pure a}}$ or $\varepsilon \notin \lang{g}$.
\end{lemma}
\begin{proofsketch}
  See Appendix~\ref{appendix:norm-null}.
\end{proofsketch}

Note that  the derivative of a simplified generator is
  simplified. This follows simply from the
  definition, since we only use smart constructors
  and parts of the original generator
  to build the derivative generators. By induction,
  this also means that repeated derivatives
  preserve simplification.

Besides clearing up the issue with {\sf Pair},
Lemma~\ref{lem:norm-null} also says that we can
define \intro(GEN){nullability} for free
generators simply as:\\
\begin{minipage}[c]{.95\textwidth}
\begin{lstlisting}
  $\gennullable$ :: FGen a -> Set a
  $\gennullable$(Pure v) = {v}
  $\gennullable$g        = $\varnothing$     (g $\neq$ Pure v)
\end{lstlisting}
\end{minipage}
Note that we get a bit more information here than
we do from regular expression nullability. For a
regular expression $r$, $\nullable r$ is either
$\varnothing$ or $\varepsilon$. Here, we allow the
null check to return either $\varnothing$ or the
singleton set containing the value in the {\sf
Pure} node. This means that $\gennullable$ for
free generators extracts a value that can be
obtained by making no further choices.

Our definition
of derivative acts the way we expect:
\begin{theorem}[Language Consistency]\label{thm:deriv-consistent}
  For all \kl[simplified form]{simplified}
  \kl[free generator]{free generators} $g$ and choices $c$,
  \[ \deriv_{c} \lang{g} = \lang{\genderiv_{c}g}. \]
\end{theorem}
\begin{proofsketch}
  By induction (see
  Appendix~\ref{appendix:deriv-consistent}).
\end{proofsketch}
This theorem says that the derivative of a free
generator's language is the same as the language
of its derivative.

Besides consistency with respect to the language
interpretation, the derivative operation should preserve the
generator output for a given sequence of choices.
If a free generator chooses
\begin{center}
  \cc{ntll} to yield \lstinline{Node True Leaf Leaf},
\end{center}
we would like for the derivative of that free
generator with respect to \cc{n} to produce the
same value after choosing \cc{tll}. We can
formalize this expectation via the parser interpretation:
\begin{theorem}[Value Consistency]\label{thm:deriv-value-consistent}
  For all \kl[simplified form]{simplified}
  \kl[free generator]{free generators} $g$, choice sequences $s$, and choices
  $c$,
  \[ \parser{\genderiv_{c}{g}}~s = \parser{g}~(c \cdot s). \]
\end{theorem}
\begin{proofsketch}
  Mostly straightforward
  induction. See
  Appendix~\ref{appendix:deriv-value-consistent}.
\end{proofsketch}
The upshot of this theorem is that derivatives do
not fundamentally change the results of a free
generator, they only fix a particular choice.

These two consistency theorems together mean that
we can simulate a free generator's choices by
taking repeated derivatives. Each derivative fixes
a particular choice, so a sequence of derivatives
fixes a choice sequence.

\section{Generating Valid Results with Gradients}\label{sec:grad}
We now put the theory of
\kl[FGen]{free generators} and their
\kl(GEN)[derivative]{derivatives} into practice.
We introduce \fullalgo{}
(\shortalgo{}), an algorithm for generating data
that satisfies a validity condition.

\SUBSECTION{The Algorithm}
\AP{} Given a simple free generator,
\intro[cgs]{\fullalgo{}}
``previews'' its choices using derivatives. In fact, it
previews all possible choices, essentially taking
the {\em gradient\/} of the free generator. (This
is akin to the gradient in calculus, which is a
vector of partial derivatives with respect to each
variable.) We write
  \[
\nabla g = \langle \genderiv_{\cc{a}}g,\ \genderiv_{\cc{b}}g,\ \genderiv_{\cc{c}}g  \rangle
\]
for the gradient of $g$ with respect to alphabet
$\{\cc{a}, \cc{b}, \cc{c}\}$. Each derivative in
the gradient can then be sampled\bcp{??}, using
$\valueDist{\cdot}$, to get a sense of how good or
bad the respective choice was. This provides a
metric that guides the algorithm toward valid inputs.

With this intuition in mind, we present the
\shortalgo{} algorithm, shown in
Figure~\ref{alg:gen-valid}, which searches for
valid results using repeated free generator
gradients.  \bcp{The little triangles are not a very standard comma syntax...
 }\hg{It's the default for this environment and I was just reading a paper that
   used them... If I have time I'll try to change the macro}

\begin{figure}[h]
  \begin{algorithmic}[1]
    \State{$g \gets G$}
    \State{$\mathcal{V} \gets \varnothing$}
    \While{$\textbf{true}$}
    \If{$\gennullable g \neq \varnothing$} $\textbf{return}~\nu g \cup \mathcal{V}$ \EndIf{}
    \If{$g = \textsf{Void}$} $g \gets G$ \EndIf{}
    \State{$\nabla g \gets \langle \genderiv_{c}g \mid c \in C \rangle $}\Comment{$\nabla g$ is the
      gradient of $g$}
    \For{$\genderiv_{c}g \in \nabla g$}
    \If{$\genderiv_{c}g = \textsf{Void}$}
    \State{$V \gets \varnothing$}
    \Else{}
    \State{$x_1, \dots, x_N \leftsquigarrow \valueDist{\genderiv_{c}g}$}\Comment{Sample $\valueDist{\genderiv_{c}g}$}
    \State{$V \gets \{x_j \mid
      \varphi(x_j)\}$}
    \EndIf{}
    \State{$f_{c} \gets |V|$}\Comment{$f_{c}$ is the {\em fitness\/} of c}
    \State{$\mathcal{V} \gets \mathcal{V} \cup V$}
    \EndFor{}
    \If{$\max_{c \in C} f_{c} = 0$} \For{$c \in C$} $f_{c} \gets 1$ \EndFor{} \EndIf{}
    \State{$g \leftsquigarrow \textsf{weightedChoice}~\{(f_{c}, \genderiv_{c}g) \mid c \in C \}$}
    \EndWhile{}
  \end{algorithmic}
  \caption{\fullalgo{}:
    Given a free generator $G$ in \kl{simplified
      form}, a sample rate constant $N$,
    and a validity predicate $\varphi$,
    this algorithm produces a set of outputs
    that all satisfy $\varphi(x)$.
  }\label{alg:gen-valid}
\end{figure}

The intuition from earlier\bcp{where?} plays out in lines
7--14, and is shown pictorially in
Figure~\ref{fig:generation-algorithm}. We take the
gradient of $g$ by taking the derivative with
respect to each possible choice, in this case
\cc{a}, \cc{b}, and \cc{c}. Then we evaluate each
of the derivatives by interpreting the free generator
with $\valueDist{\cdot}$, sampling values from the
resulting generator, and counting how many of
those results are valid with respect to $\varphi$.
The precise number of samples is controlled by
$N$, the sample rate constant; this is up to the
user, but in general higher values for $N$ will
give better information about each derivative at
the expense of time spent sampling. At the end of
sampling, we have values $f_{\cc{a}}$,
$f_{\cc{b}}$, and $f_{\cc{c}}$, which we can think
of as the ``fitness'' of each choice. We then pick
a choice randomly, weighted based on fitness, and
continue until our choices produce a valid output.

\begin{figure}[h]
  \centering
  \includegraphics[width=0.8\columnwidth]{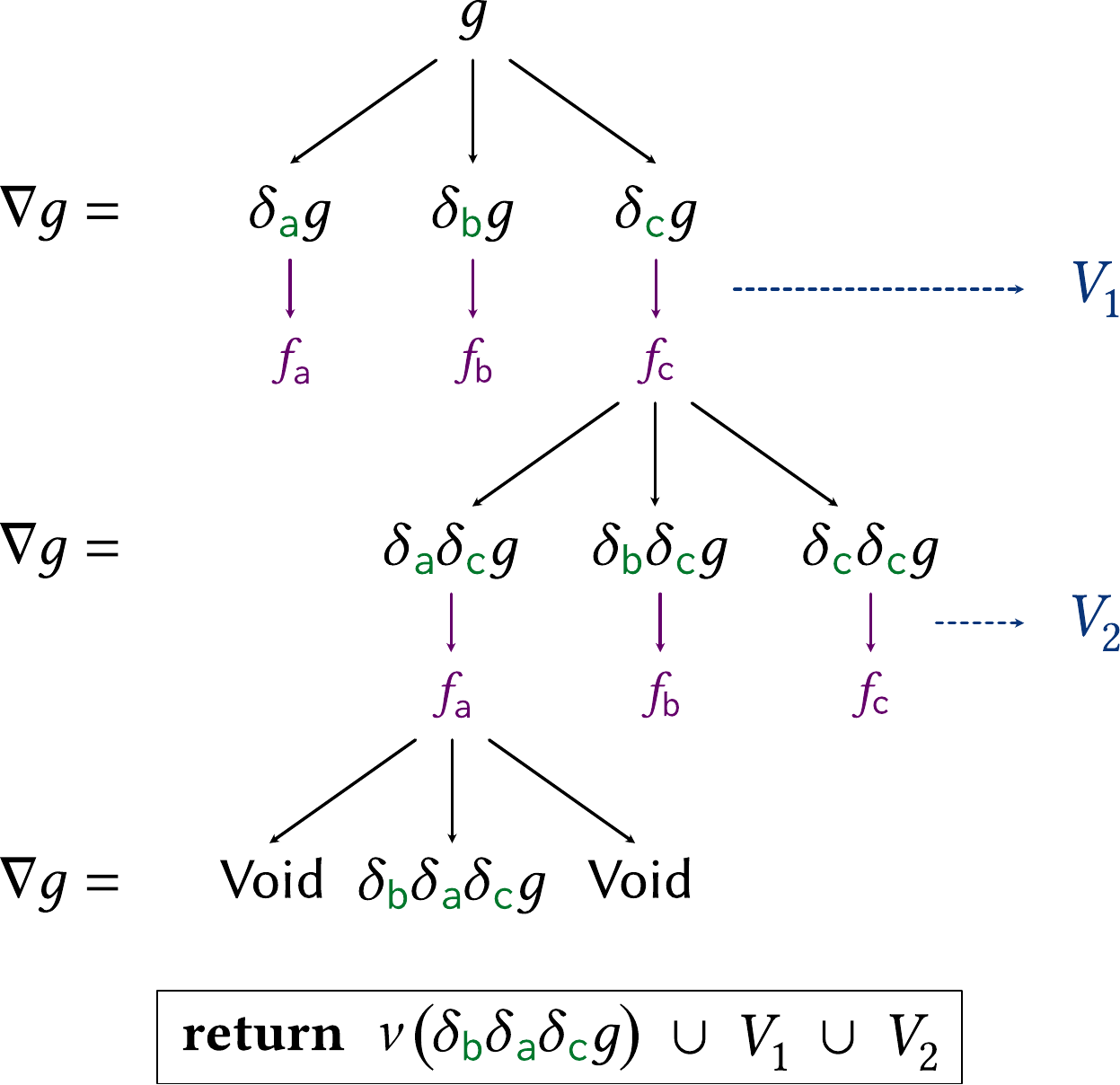}
  \caption{The main loop of \fullalgo{}.}\label{fig:generation-algorithm}
\end{figure}

Critically, we avoid wasting effort by saving the
samples ($\mathcal{V}$) that we use to evaluate
the gradients. Many of those samples will be valid
results that we can use, so there is no reason to
throw them away.

\SUBSECTION{Modified Distributions}\label{subsec:modified-dist}
Interestingly, this algorithm works equally well
for free generators whose distributions have been
replaced (as discussed in
\S\ref{subsec:replacing-dist}). Recall that we
modify the distribution of a free generator $g$ by
pairing it with a pair of a history $h$ and a
distribution function $d$. We can define the
derivative of such a structure to be:
\begin{lstlisting}
  $\genderiv_{c}$((h, d), g) = ((h $\cdot$ c, d), $\genderiv_{c}$g)
\end{lstlisting}
We take the derivative of $g$ and internalize $c$
into the distribution's history. Furthermore, we
can say that a modified generator's nullable set
is the same as the nullable set of the underlying
generator.

These definitions, along with the ones in
\S\ref{subsec:replacing-dist}, are enough to
replicate \fullalgo{} for a free generator with an
external distribution.

\section{Exploratory Evaluation}\label{sec:evaluation}
This paper is primarily about the theory of
\kl[free generator]{free generators} and their
\kl(GEN)[derivative]{derivatives}, but readers may be curious (as we were)
to see how well \kl[cgs]{\fullalgo{}}
performs on a few property-based testing
benchmarks.  In this section, we describe some preliminary experiments in this
direction; the results suggest that, with a few interesting caveats,
\shortalgo{} is a promising approach to the
valid generation problem.

\SUBSECTION{Experimental Setup}
Our experiments explore how well \shortalgo{}
improves on the base, na\"\i ve generator by
comparing it to {\em rejection
sampling}, which takes a na\"\i ve
generator, samples from it, and discards
any results that are not valid. Rejection sampling
is the
default method that {\sc QuickCheck} uses for
properties with preconditions when no bespoke
generator is available, and makes for a clean
baseline that we can compare \shortalgo{} to.

We use four simple free generators to test four
different benchmarks: \bench{BST}, \bench{SORTED},
\bench{AVL}, and \bench{STLC}. Details about
each of these benchmarks are given in
Table~\ref{table:benchmark-detail}.

\begin{table*}[h]
  \begin{tabular}{lllrr}
    & Free Generator & Validity Condition & \multicolumn{1}{l}{$N$} & \multicolumn{1}{l}{Depth} \\\midrule
    \bench{BST} & Binary trees with values 0--9 & Is a valid BST & 50 & 5 \\
    \bench{SORTED} & Lists with values 0--9 & Is sorted & 50 & 20 \\
    \bench{AVL} & Binary trees with values and stored heights 0--9 & Is a valid AVL tree (balanced) & 500 & 5 \\
    \bench{STLC} & Arbitrary ASTs for $\lambda$-terms & Is well-typed & 400 & 5
  \end{tabular}
  \caption{Overview of benchmarks.}\label{table:benchmark-detail}
\end{table*}

Each of our benchmarks requires a simple free
generator to act as a baseline and as a starting
point for \shortalgo{}. For consistency, and to
avoid potential biases, our generators follow the
respective inductive data types as closely as
possible. For example, {\sf fgenTree}, shown in
\S\ref{sec:free} and used in the \bench{BST}
benchmark, follows the structure of {\sf Tree}
exactly. We chose values for $N$ via trial and
error in order to balance fitness accuracy with
sampling time.

\SUBSECTION{Results}
We ran \shortalgo{} and {\sc Rejection} on each
benchmark for one minute (on a MacBook Pro with an
M1 processor and 16GB RAM) and recorded the
unique valid values produced. We counted unique
values because duplicate tests are generally less useful than fresh ones (in
property-based testing of pure programs, in particular, duplicate tests add no
value). The totals, averaged over 10 trials, are
presented in Table~\ref{table:unique-over-time}.

\begin{table}[h]
  \begin{tabular}{lrrrr}
    & \multicolumn{1}{l}{\bench{BST}} & \multicolumn{1}{l}{\bench{SORTED}} & \multicolumn{1}{l}{\bench{AVL}} & \multicolumn{1}{l}{\bench{STLC}} \\\midrule
    {\sc Rej.} & \num{9729} (\num{103}) & \num{6587} (\num{125}) & \num{156} (\num{5}) & \num{105602} (\num{2501}) \\
    \shortalgo{} & \num{22349} (\num{416}) & \num{58656} (\num{881}) & \num{220} (\num{1}) & \num{297703} (\num{11726})
  \end{tabular}
  \caption{Unique valid values generated in 60
    seconds ($n = 10$ trials).}\label{table:unique-over-time}
\end{table}

These measurements show that \shortalgo{}
is always able to generate more unique values than
{\sc Rejection} in the same amount of time, and it
often generates {\em significantly\/} more. (The exception is the \bench{AVL}
benchmark; we discuss this below.)

Besides unique values, we measured some
other metrics; the charts in
Figure~\ref{fig:stlc-stats} give some deeper
insights for the \bench{STLC} benchmark. The first plot (``Unique Terms over
Time'') shows that, after one minute, \shortalgo{}
has not yet begun to ``run out'' of unique terms to
generate. Additionally ``Normalized Size
Distribution'' chart shows that \shortalgo{} also
generates larger terms on average. This is good
from the perspective of property-based testing,
where test size is often positively correlated
with bug-finding power, since larger test inputs
tend to exercise more of the implementation code.
Charts for the remaining benchmarks are in
Appendix~\ref{appendix:full-results}.

\begin{figure}[h]
  \centering
  \includegraphics[width=.95\columnwidth]{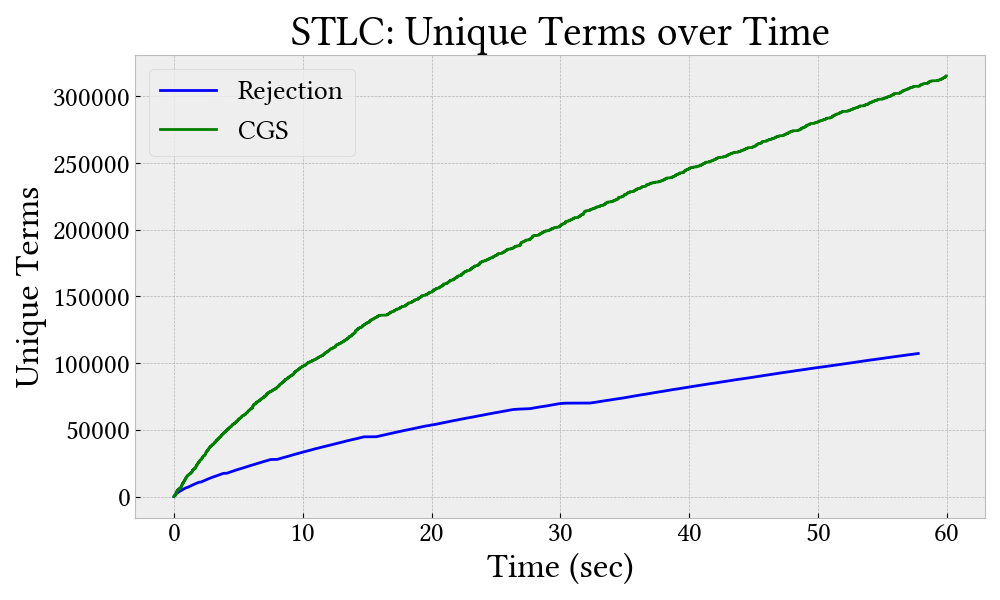}
  \includegraphics[width=.95\columnwidth]{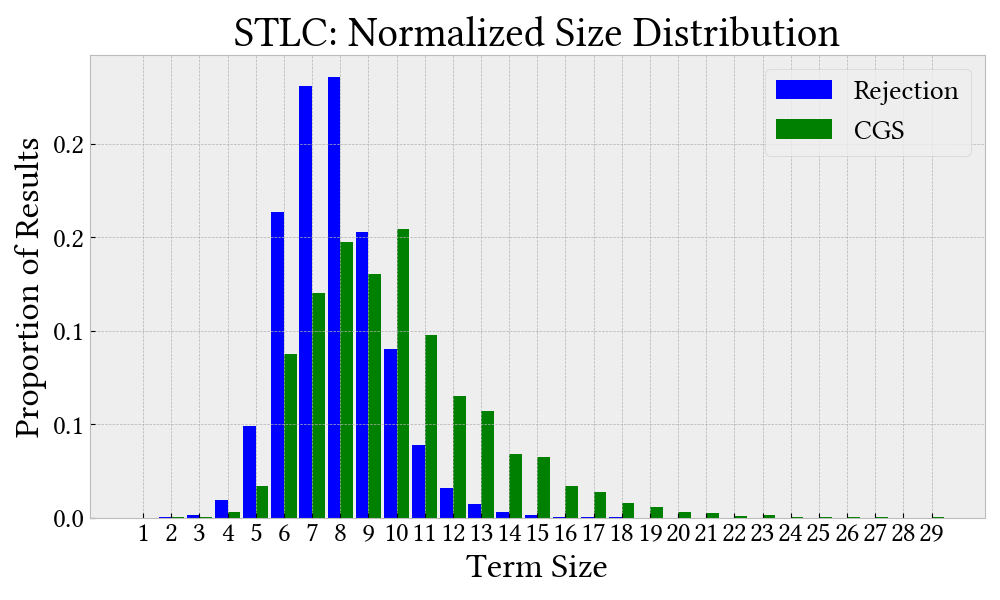}
  \caption{Unique values and term sizes for the
    \bench{STLC} benchmark (first trial).}\label{fig:stlc-stats}
\end{figure}

\SUBSECTION{Measuring Diversity}
When testing, we care about more than
just the number of valid test inputs generated in
a period of time---we care about the {\em
diversity\/} of those inputs, since a more diverse test
suite will find more bugs more quickly.

Our {\em diversity metric\/} relies on
the fact that each value is roughly isomorphic to
the choice sequence that generated it. For
example, in the case of \bench{BST}, the sequence
\cc{n5l6ll} can be parsed to produce
\lstinline{Node 5 Leaf (Node 6 Leaf Leaf)} and a
simple in-order traversal can recover \cc{n5l6ll}
again. Thus,
choice sequence diversity is a reasonable proxy for value
diversity.

We estimated the average Levenshtein
distance~\cite{levenshtein1966binary} (the number
of edits needed to turn one string into another) between
pairs of choice sequences in the values generated
by each of our algorithms. Computing an exact mean distance between all pairs in such a large set
would be very expensive, so we settled for the mean of a
random sample of 3000 pairs from each set of valid
values. The results are summarized in
Table~\ref{table:lev-distance}.

\begin{table}[h]
  \begin{tabular}{lrrrr}
    & \multicolumn{1}{l}{\bench{BST}} & \multicolumn{1}{l}{\bench{SORTED}} & \multicolumn{1}{l}{\bench{AVL}} & \multicolumn{1}{l}{\bench{STLC}} \\\midrule
    {\sc Rej.} & $7.70 (1.71)$ & $4.80 (1.15)$ & $4.42 (2.01)$ & $12.24 (4.55)$ \\
    \shortalgo{} & $8.89 (1.95)$ & $7.28 (1.92)$ & $4.35 (1.98)$ & $13.62 (4.72)$
  \end{tabular}
  \caption{Average Levenshtein distance between
pairs of choice
sequences (first trial).}\label{table:lev-distance}
\end{table}

While \bench{SORTED} does see significantly
improved diversity, the effect is less dramatic
\bench{STLC} and \bench{BST}, and diversity for
\bench{AVL} actually gets slightly worse.

One explanation for these lackluster results rests
on the way \shortalgo{} retains intermediate
samples. While the first few samples will be
mostly uncorrelated, the samples drawn later on in
the generation process (once a number of choices
have been fixed) will tend to be similar to one
another. This likely results in some clusters of
inputs that are all valid but that only explore
one particular shape of input.

Of course, is already common practice to test
clusters of similar inputs in certain fuzzing
contexts~\cite{DBLP:journals/pacmpl/Lampropoulos0P19},
so the fact that \shortalgo{} does this is not
unusual. In fact, this method has been shown to be
effective at finding bugs in some cases.
Additionally, for most of our benchmarks (again,
we return to \bench{AVL} in a moment) \shortalgo{}
does increase diversity of tests; combined with
the sheer number of valid inputs available, this
means that \shortalgo{} covers a slightly larger
space of tests much more thoroughly. This effect
should lead to better bug-finding in testing
scenarios.

\SUBSECTION{The Problem with \bench{AVL}}\label{subsec:problem-avl}
The \bench{AVL} benchmark is an outlier in most of
these measurements: \shortalgo{} only manages to
find a modest number of extra valid AVL trees, and
their pairwise diversity is actually slightly
worse than that of rejection sampling. Why might
this be? We suspect that this effect arises
because AVL trees are quite difficult to find
randomly. Balanced binary search trees are hard to
generate on their own, and AVL trees are even more
difficult because the generator must guess the
correct height to cache at each node. This is why
rejection sampling only finds $\num{156}$ AVL
trees in the time it takes to find $\num{9762}$
binary search trees.

This all means that \shortalgo{} is unlikely find
{\em any\/} valid trees while sampling. In
particular, the check in line 15 of
Figure~\ref{fig:generation-algorithm} will often
be true, meaning that choices will be made
uniformly at random rather than guided by the
fitness of the appropriate derivatives. We could
reduce this effect by significantly increasing the
sample rate constant $N$, but then sampling time
would likely dominate generation time, resulting
in worse performance overall.

The lesson here seems to be that the \shortalgo{}
algorithm does not work well with especially
hard-to-satisfy predicates. In \S\ref{sec:future},
we present an idea that would do some of the hard
work ahead of time and help with this issue, but
clearly many predicates (including complex ones
like well-typedness of STLC terms) are within
reach of the current algorithm. Indeed, as long as
every $N$ samples from the na\"\i ve generator has
at least a few valid values on average, the
\bench{AVL} issue will not come up. We expect that
many real-world structural and semantic
constraints will require a small enough $N$ for
\shortalgo{} to be effective.

\section{Related Work}\label{sec:related}
We discuss a number of approaches that are similar
to ours, via either connections to free generators
or connections to our \fullalgo{} algorithm.

\SUBSECTION{Parsing and Generation}
The connection between parsers and generators is
not just ``intellectual folklore''---it is used in
some implementations too. At least two popular
property-based testing libraries, {\sc
Hypothesis}~\cite{maciver2019hypothesis} and {\sc
Crowbar}~\cite{dolan2017testing}, implement
generators by parsing a stream of random bits (and
there may very well be others that we do not know
of.) This further illustrates the value in
formalizing the connection between parsers and
generators, as a way to explain existing
implementations uncover potential opportunities.

The {\sc Clotho}~\cite{darragh2021clotho} library
introduces ``parametric randomness,'' providing a
way to carefully control generator choices from
outside of the generator. While
\citeauthor{darragh2021clotho} do not use parsing
in their formalism, it is still exciting to see
others considering the implications of controlling
generator choices externally.

\SUBSECTION{Free Applicative Generators}
\citet{ClaessenDP15} present a generator
representation that is structurally similar to our
free generators, but which is used in a very
different way. They primarily use the syntactic
structure of their generators (they call them
``spaces'') to control the size distribution of
generated outputs; in particular, spaces do not
make choice information explicit in the way free
generators do. \citeauthor{ClaessenDP15}'s
generation approach uses {\sc Haskell}'s laziness,
rather than derivatives and sampling, to prune
unhelpful paths in the generation process. This
pruning procedure performs well when validity
conditions are written to take advantage of
laziness, but it is highly dependent on evaluation
order and it does not differentiate between
generator choices that are not obviously bad. In
contrast, \shortalgo{} respects observational
equivalence between predicates and uses sampling
to weight next choices.

\SUBSECTION{The Valid Generation Problem}
Many other approaches to the valid generation
problem have been explored.

The domain-specific language for generators provided by the {\sc
QuickCheck} library~\cite{hughes2007quickcheck} makes it
easier to write manual generators that produce
valid inputs by construction. This approach is
extremely general, but it can be labor intensive.
In the present work, we avoid manual techniques like this in the hopes
of making property-based testing
more accessible to programmers that do not have the time or expertise to
write their own custom generators.

The {\sc Luck}~\cite{LuckPOPL} language provides a sort of middle-ground
solution; users are still required to put in some effort, but they are able to
define generators and validity predicates at the same time. {\sc Luck} provides
a satisfying solution if users are starting from scratch and willing to learn a
domain-specific language, but if validity predicates have already been written
or users do not want to learn a new language, a more automated
solution is preferable.

When validity predicates are expressed as
inductive relations, approaches like the one in
{\em Generating Good Generators for Inductive
Relations\/}~\cite{lampropoulos2017generating} are
extremely powerful. Unfortunately, most
programming languages cannot express inductive
relations that capture the kinds of preconditions
that we care about.

{\sc Target}~\cite{loscher2017targetedpbt} uses
search strategies like hill-climbing and simulated
annealing to supplement random generation and
significantly streamline property-based testing.
\citeauthor{loscher2017targetedpbt}'s approach
works extremely well when inputs have a sensible
notion of ``utility,'' but in the case of valid
generation the utility is often degenerate---0 if
the input is invalid, and 1 if it is valid---with
no good way to say if an input is ``better'' or
``worse.'' In these cases, derivative-based
searches may make more sense.

Some approaches use machine learning to
automatically generate valid inputs. {\sc
Learn\&Fuzz}~\cite{godefroid2017learn} generates
valid data using a recurrent neural network. While
the results are promising, this solution seems to
work best when a large corpus of inputs is already
available and the validity condition is more
structural than semantic. In the same vein, {\sc
RLCheck}~\cite{DBLP:conf/icse/ReddyLPS20} uses
reinforcement learning to guide a generator to
valid inputs. This approach served as early
inspiration for our work, and we think that
the theoretical advance of generator derivatives
may lead improved learning algorithms in the
future (see \S\ref{sec:future}).

\section{Future Directions}\label{sec:future}
There are a number of exciting paths forward from
this work; some continue our theoretical
exploration and others look towards algorithmic
improvements.

\SUBSECTION{Bidirectional Free Generators}
We believe that we have only scratched the surface
of what is possible with \kl[free generator]{free
generators}. One concrete next step is to merge
the theory of free generators with the emerging
theory of {\em
ungenerators\/}~\cite{goldstein2021ungenerators}.
\citeauthor{goldstein2021ungenerators} expresses
generators that can be run both forward (to
generate values as usual) and {\em backward}. In
the backward direction, the program takes a value
that the generator might have generated and
``un-generates'' it to give a sequence of choices
that the generator might have made when generating
that value.

Free generators are quite compatible with these
ideas, and turning a free generator into a
bidirectional generator that can both generate and
ungenerate should be fairly straightforward. From
there, we can build on the ideas in the
ungenerators work and use the backward direction
of the generator to learn a distribution of
choices that approximates some user-provided
samples of ``desirable'' values. Used in
conjunction with the extended algorithm from
\S\ref{subsec:modified-dist}, this would give a
better starting point for generation with little
extra work from the user.

\SUBSECTION{Algorithmic Optimizations}
In \S\ref{subsec:problem-avl}, we saw some
problems with the \kl[cgs]{\fullalgo{}} algorithm:
because \shortalgo{} evaluates
derivatives via sampling, it does poorly when
validity conditions are particularly difficult to
satisfy. This begs the question: might it be
possible to evaluate the fitness of a derivative
without na\"\i vely sampling?

One potential angle involves staging the sampling
process. Given a free generator with a depth
parameter, we can first evaluate choices on
generators of size 1, then evaluate choices with
size 2, etc. These intermediate stages would make
gradient sampling more successful at larger sizes,
and might significantly improve the results on
benchmarks like \bench{AVL}. Unfortunately, this
kind of approach might perform poorly on
benchmarks like \bench{STLC} where the validity
condition is not uniform: size-1 generators would
avoid generating variables, leading larger
generators to avoid variables as well. In any
case, we think this design space is worth
exploring.

\SUBSECTION{Making Choices with Neural Networks}
Another algorithmic optimization is a bit farther
afield: we think it may be possible to use
recurrent neural networks (RNNs) to improve our
generation procedure.

As \fullalgo{} makes choices, it generates useful
data about the frequencies with which choices
should be made. Specifically, every iteration of
the algorithm produces a pair of a history and a
distribution over next choices that looks
something like
\[
  \cc{abcca} \mapsto \{ \cc{a}: 0.3, \cc{b}: 0.7, \cc{c}: 0.0 \} .
\]
In the course of \shortalgo{}, this information is
used once (to make the next choice) and then
forgotten---what if there was a way to learn from
it? Pairs like this could be used to train an RNN
to make choices that are similar to the ones made
by \shortalgo{}.

There are still details to work out, including
network architecture, hyper-parameters, etc., but
in theory we could run \shortalgo{} for a while,
then train the model, and after that point only
use the RNN to generate valid data. Setting things
up this way would recover some of the time that is
currently wasted by the constant sampling of
derivative generators.

One could imagine a user writing a definition of a
type and a predicate for that type, and then
setting the model to train while they work on
their algorithm. By the time the algorithm is
finished and ready to test, the RNN model would be
trained and ready to produce valid test inputs. A
workflow like this could significantly increase
adoption of property-based testing in industry. \\

Free generators and their derivatives are powerful
structures that give a unique and flexible
perspective on random generation. Our formalism
yields a useful algorithm and clarifies the
folklore that a generator is a parser of
randomness.


\bibliography{references}

\clearpage
\appendix
\onecolumn
\noindent {\LARGE \bf Appendix}
\section{Proof of Lemma~\ref{lem:norm-null}}\label{appendix:norm-null}
\begin{repeat-lemma}{\ref{lem:norm-null}}
  If a \kl{free generator} $g$ is in \kl{simplified form}, then either $g =
  \text{\lstinline{Pure a}}$ or $\varepsilon \notin \lang{g}$.
\end{repeat-lemma}
\begin{proof}
  We proceed by induction on the structure of $g$.
  \begin{itemize}
  \item[Case] $g = \textsf{Void}$. Trivial.
  \item[Case] $g = \textsf{Pure a}$. Trivial.
  \item[Case] $g = \textsf{Pair x y}$.
    \noindent By our inductive hypothesis, $x = \textsf{Pure a}$ or $\varepsilon
    \notin \lang{\textsf{x}}$.

    \noindent Since the smart constructor $\otimes$ never constructs a $\textsf{Pair}$ with
    $\textsf{Pure}$ on the left, it must be that $\varepsilon \notin \lang{\textsf{x}}$.

    \noindent Therefore, it must be the case that $\varepsilon \notin \lang{\textsf{Pair x y}}$.
  \item[Case] $g = \textsf{Map f x}$.

    \noindent Similarly to the previous case, our inductive hypothesis and simplification assumptions
    imply that $\varepsilon \notin \lang{\textsf{x}}$.

    \noindent Therefore, $\varepsilon \notin \lang{\textsf{Map f y}}$.
  \item[Case] $g = \textsf{Select xs}$.

    \noindent It is always the case that $\varepsilon \notin \lang{\textsf{Select xs}}$.
  \end{itemize}
  Thus, we have shown that every simplified free generator is either \lstinline{Pure a} or
  cannot accept the empty string.
\end{proof}

\clearpage
\section{Proof of Theorem~\ref{thm:coherence}}\label{appendix:coherence}
\begin{lemma}\label{lem:parser-gen-pair}
  Pairing two \kl[parser interpretation]{parser
    interpretations} and mapping over the
  concatenation of the associated choice
  distributions is equal to a function of the two
  parsers mapped over the distributions
  individually. Specifically, for all
  \kl[simplified form]{simplified}
  \kl[free generator]{free generators} {\sf x} and {\sf y},
  \begin{align*}
    ((\lambda x\ y \to (x, y)) \fmap \parser{\textsf{x}} \ap \parser{\textsf{y}}) \fmap ((\cdot) \fmap \choiceDist{\textsf{x}} \ap \choiceDist{\textsf{y}})
    &\gequiv (\lambda a_\bot\ b_\bot \to \textbf{\textsf{case}}~(a_\bot, b_\bot)~\textbf{\textsf{of}} \\
    &\quad\quad(\textsf{Just}~(a, \_), \textsf{Just}~(b, \_)) \to \textsf{Just}~((a, b), \varepsilon) \\
    &\quad\quad \_ \to \textsf{Nothing}) \\
    &\quad \fmap (\parser{\textsf{x}} \fmap \choiceDist{\textsf{x}}) \ap (\parser{\textsf{y}} \fmap \choiceDist{\textsf{y}})
  \end{align*}
\end{lemma}
\begin{proof}
  First, note that for any simplified generator, $g$, if $\choiceDist{g}$ generates a string $s$,
  for any other string $t$ $\parser{g}~(s \cdot t) = \textsf{Just}~(a, t)$ for some value $a$. This
  can be shown by induction on the structure of $g$.

  \noindent Now, assume $\choiceDist{\textsf{x}}$ generates a string $s$, and
  $\choiceDist{\textsf{y}}$ generates $t$. This means that $(\cdot) \fmap \choiceDist{\textsf{x}} \ap
  \choiceDist{\textsf{y}}$ generates $s \cdot t$.

  \noindent By the above fact, it is simple to show that both sides of the above equation simplify
  to $\textsf{Just}~((a, b), \varepsilon)$ for some values $a$ and $b$ that depend on the particular
  interpretations of {\sf x} and {\sf y}.

  \noindent Since this is true for any $s$ and $t$ that the choice distributions generate, the
  desired fact holds.
\end{proof}

\begin{repeat-theorem}{\ref{thm:coherence}}
  Every \kl[simplified form]{simplified} \kl{free
    generator} can be factored into a parser and a
  distribution over choice sequences. In other
  words, for all simplified free generators $g \neq
  \textsf{Void}$,
  \[
    \parser{g} \fmap \choiceDist{g} \gequiv (\lambda x \to \textsf{Just}~(x, \varepsilon)) \fmap
    \valueDist{g}.
  \]
\end{repeat-theorem}
\begin{proof}
  We proceed by induction on the structure of $g$.
  \begin{itemize}
  \item[Case] $g = \textsf{Pure a}$.
    \begin{align*}
      \parser{\textsf{Pure a}} \fmap \choiceDist{\textsf{Pure a}} &\gequiv \textsf{pure}~(\textsf{Just}~(\textsf{a}, \varepsilon))
      && \text{(by defn)} \\
                                                                 &\gequiv (\lambda x \to \textsf{Just}~(x, \varepsilon))
                                                                   \fmap \valueDist{\textsf{Pure a}}
      && \text{(by defn)}
    \end{align*}
  \item[Case] $g = \textsf{Pair x y}$.
    \begin{align*}
      \parser{\textsf{Pair x y}} \fmap \choiceDist{\textsf{Pair x y}} &\gequiv ((\lambda x\ y \to (x, y)) \fmap \parser{\textsf{x}} \ap \parser{\textsf{y}}) \fmap ((\cdot) \fmap \choiceDist{\textsf{x}} \ap \choiceDist{\textsf{y}})
      && \text{(by defn)} \\
                                                                      &\gequiv (\lambda a_\bot\ b_\bot \to \textbf{\textsf{case}}~(a_\bot, b_\bot)~\textbf{\textsf{of}} \\
                                                                      &\quad\quad(\textsf{Just}~(a, \_), \textsf{Just}~(b, \_)) \to \textsf{Just}~((a, b), \varepsilon) \\
                                                                      &\quad\quad \_ \to \textsf{Nothing}) \\
                                                                      &\quad \fmap (\parser{\textsf{x}} \fmap \choiceDist{\textsf{x}}) \ap (\parser{\textsf{y}} \fmap \choiceDist{\textsf{y}})
      && \text{(by Lemma~\ref{lem:parser-gen-pair})} \\
                                                                      &\gequiv (\lambda a_\bot\ b_\bot \to \textbf{\textsf{case}}~(a_\bot, b_\bot)~\textbf{\textsf{of}} \\
                                                                      &\quad\quad(\textsf{Just}~(a, \_), \textsf{Just}~(b, \_)) \to \textsf{Just}~((a, b), \varepsilon) \\
                                                                      &\quad\quad \_ \to \textsf{Nothing}) \\
                                                                      &\quad \fmap ((\lambda x \to \textsf{Just}~(x, \varepsilon)) \fmap \valueDist{\textsf{x}}) \ap ((\lambda x \to \textsf{Just}~(x, \varepsilon)) \fmap \valueDist{\textsf{y}})
      && \text{(by IH)} \\
                                                                      &\gequiv (\lambda x \to \textsf{Just}~(x, \varepsilon)) \fmap ((\lambda x\ y \to (x, y)) \fmap \valueDist{\textsf{x}} \ap \valueDist{\textsf{y}})
      && \text{(by app.\ properties)} \\
                                                                      &\gequiv (\lambda x \to \textsf{Just}~(x, \varepsilon)) \fmap \valueDist{\textsf{Pair x y}}
      && \text{(by defn)}
    \end{align*}
  \item[Case] $g = \textsf{Map f x}$.
    \begin{align*}
      \parser{\textsf{Map f x}} \fmap \choiceDist{\textsf{Map f x}} &\gequiv (\textsf{f} \fmap \parser{\textsf{x}}) \fmap \choiceDist{x}
      && \text{(by defn)} \\
                                                                    &\gequiv \textsf{f} \fmap (\parser{\textsf{x}} \fmap \choiceDist{x})
      && \text{(by functor properties)} \\
                                                                    &\gequiv \textsf{f} \fmap ((\lambda x \to \textsf{Just}~(x, \varepsilon)) \fmap \valueDist{\textsf{x}})
      && \text{(by IH)} \\
                                                                    &\gequiv (\lambda x \to \textsf{Just}~(x, \varepsilon)) \fmap (\textsf{f} \fmap \valueDist{\textsf{x}})
      && \text{(by functor properties)} \\
                                                                      &\gequiv (\lambda x \to \textsf{Just}~(x, \varepsilon))
                                                                        \fmap \valueDist{\textsf{Map f x}}
      && \text{(by defn)}
    \end{align*}
  \item[Case] $g = \textsf{Select xs}$.
    \begin{align*}
      \parser{\textsf{Select xs}} \fmap \choiceDist{\textsf{Select xs}} &\gequiv (\textsf{choice}~(\textsf{map}~(\lambda (c, \textsf{x}) \to (c, \parser{\textsf{x}}))~\textsf{xs})) \\
                                                                        &\quad \fmap \textsf{oneof}~(\textsf{map}~(\lambda (c, \textsf{x}) \to (c \cdot) \fmap \choiceDist{\textsf{x}})~\textsf{xs})
                                                                        && \text{(by defn)} \\
                                                                        &\gequiv \textsf{oneof}~(\textsf{map}~(\lambda (c, \textsf{x}) \to \\
                                                                        &\quad\quad (\textsf{choice}~(\textsf{map}~(\lambda (c, \textsf{x}) \to (c, \parser{\textsf{x}}))~\textsf{xs})) \circ (c \cdot) \fmap \choiceDist{\textsf{x}} \\
                                                                        &\quad)~\textsf{xs})
                                                                        && \text{(by generator properties)} \\
                                                                        &\gequiv \textsf{oneof}~(\textsf{map}~(\lambda (\_, \textsf{x}) \to \parser{\textsf{x}} \fmap \choiceDist{\textsf{x}})~\textsf{xs})
                                                                        && \text{(by parser properties)} \\
                                                                        &\gequiv \textsf{oneof}~(\textsf{map}~(\lambda (\_, \textsf{x}) \to (\lambda x \to \textsf{Just}~(x, \varepsilon)) \fmap \valueDist{\textsf{x}})~\textsf{xs})
                                                                        && \text{(by IH)} \\
                                                                        &\gequiv (\lambda x \to \textsf{Just}~(x, \varepsilon)) \fmap \textsf{oneof}~(\textsf{map}~(\lambda (\_, \textsf{x}) \to \valueDist{\textsf{x}})~\textsf{xs})
                                                                        && \text{(by generator properties)} \\
                                                                        &\gequiv (\lambda x \to \textsf{Just}~(x, \varepsilon)) \fmap \valueDist{\textsf{Select xs}}
                                                                        && \text{(by defn)}
    \end{align*}
  \end{itemize}
  Thus, generators can be coherently factored into a parser and a distribution.
\end{proof}

\clearpage
\section{Proof of Theorem~\ref{thm:deriv-consistent}}\label{appendix:deriv-consistent}
\begin{repeat-theorem}{\ref{thm:deriv-consistent}}
  For all \kl[simplified form]{simplified}
  \kl[free generator]{free generators} $g$ and choices $c$,
  \[ \deriv_{c} \lang{g} = \lang{\genderiv_{c}g}. \]
\end{repeat-theorem}
\begin{proof}
  We again proceed by induction on the structure of $g$.
  \begin{itemize}
    \item[Case] $g = \textsf{Void}$. $\varnothing = \varnothing$.
    \item[Case] $g = \textsf{Pure a}$. $\varnothing = \varnothing$.
    \item[Case] $g = \textsf{Pair x y}$.
      \begin{align*}
        \deriv_{c} \lang{\textsf{Pair x y}} &= \deriv_{c} (\lang{\textsf{x}} \cdot \lang{\textsf{y}})
        && \text{(by defn)} \\
                                            &= \deriv_{c} (\lang{\textsf{x}}) \cdot \lang{\textsf{y}} + \nullable \lang{\textsf{x}} \cdot \deriv_{c} \lang{\textsf{y}}
        && \text{(by defn)} \\
                                                 &= \deriv_{c} (\lang{\textsf{x}}) \cdot \lang{\textsf{y}}
        && \text{(by Lemma~\ref{lem:norm-null})} \\
                                                 &= \lang{\genderiv_{c} \textsf{x}} \cdot \lang{\textsf{y}}
        && \text{(by IH)} \\
                                                 &= \lang{\textsf{Pair}~(\genderiv_{c} \textsf{x})~\textsf{y}}
        && \text{(by defn)} \\
                                                 &= \lang{\genderiv_{c}\textsf{Pair x y}}
        && \text{(by defn)}
      \end{align*}
    \item[Case] $g = \textsf{Map f x}$.
      \begin{align*}
        \deriv_{c} \lang{\textsf{Map f x}} &= \deriv_{c} \lang{\textsf{x}}
        && \text{(by defn)} \\
                                                &= \lang{\genderiv_{c} \textsf{x}}
        && \text{(by IH)} \\
                                                &= \lang{\textsf{Map f}~(\genderiv_{c} \textsf{x})}
        && \text{(by defn)} \\
                                                &= \lang{\genderiv_{c}(\textsf{Map f x})}
        && \text{(by defn*)}
      \end{align*}
      *Note that the last step follows because \lstinline{Map f x} is assumed to be simplified, so
      $x \neq \textsf{Pure a}$. This means that \lstinline{f fmap x} is equivalent to
      \lstinline{Map f x}.
    \item[Case] $g = \textsf{Select xs}$. If there is no pair \lstinline{(c, x)} in
      \textsf{xs}, then $\varnothing = \varnothing$. Otherwise,
      \begin{align*}
        \deriv_{c} \lang{\textsf{Select xs}} &= \deriv_{c} \{c \cdot s \mid s \in \lang{\textsf{x}}\} && \text{(by defn)} \\
                                             &= \lang{\textsf{x}} && \text{(by defn)} \\
                                             &= \lang{\genderiv_{c}(\textsf{Select xs})} && \text{(by defn)}
      \end{align*}
  \end{itemize}
  Thus we have shown that the symbolic derivative
  of free generators is compatible with the
  derivative of the generator's
  language.
\end{proof}
There is another proof of this theorem, suggested
by Alexandra Silva, which uses the fact that
$2^{\Sigma^*}$ is the final coalgebra, along with
the observation that {\sf FGen} has a $2 \times
{(-)}^{\Sigma}$ coalgebraic structure. This
approach is certainly more elegant, but it
abstracts away some helpful operational intuition.

\clearpage
\section{Proof of Theorem~\ref{thm:deriv-value-consistent}}\label{appendix:deriv-value-consistent}
\begin{repeat-theorem}{\ref{thm:deriv-value-consistent}}
  For all \kl[simplified form]{simplified}
  \kl[free generator]{free generators} $g$, choice sequences $s$, and choices
  $c$,
  \[ \parser{\genderiv_{c}{g}}~s = \parser{g}~(c \cdot s). \]
\end{repeat-theorem}
\begin{proof}
  For simplicity, we prove the point-free version
  of this claim, i.e.:
  \[ \parser{\genderiv_{c}{g}} = \parser{g} \circ (c \cdot) \]
  We proceed by induction on $g$.
  \begin{itemize}
  \item[Case] $g = \textsf{Void}$. $\textsf{Nothing} = \textsf{Nothing}$.
  \item[Case] $g = \textsf{Pure a}$. $\textsf{Nothing} = \textsf{Nothing}$.
  \item[Case] $g = \textsf{Pair x y}$.
    \begin{align*}
      \parser{\genderiv_{c}{(\textsf{Pair x y})}} &= \parser{\genderiv_{c}{\textsf{x}} \otimes \textsf{y}} && \text{(by defn)}\\
                                                  &= (\lambda x~y \to (x, y)) \fmap \parser{\genderiv_{c}{\textsf{x}}} \ap \parser{\textsf{y}} && \text{(by defn)} \\
                                                  &= ((\lambda x~y \to (x, y)) \fmap \parser{\textsf{x}} \ap \parser{\textsf{y}}) \circ (c \cdot) && \text{(by IH \& Lemma~\ref{lem:norm-null}*)} \\
                                                  &= \parser{\textsf{Pair x y}}~(c \cdot) && \text{(by defn)}
    \end{align*}
    *We can use Lemma~\ref{lem:norm-null} to show
    that \textsf{x} must consume at least one
    character. Thus, we can move the $c$ in the
    derivative out into the final string, and trust
    that \textsf{x} will consume it.
  \item[Case] $g = \textsf{Map f x}$.
    \begin{align*}
      \parser{\genderiv_{c}{(\textsf{Map f x})}} &= \parser{f \fmap \genderiv_{c}{\textsf{x}}} && \text{(by defn)}\\
                                               &= f \fmap \parser{\genderiv_{c}{\textsf{x}}} && \text{(by app.\ properties)}\\
                                               &= f \fmap (\parser{\textsf{x}} \circ (c \cdot)) && \text{(by IH)}\\
                                               &= (f \fmap \parser{\textsf{x}}) \circ (c \cdot) && \text{(by app.\ properties)}\\
                                               &= \parser{\textsf{Map f x}} \circ (c \cdot) && \text{(by defn)}
    \end{align*}
  \item[Case] $g = \textsf{Select xs}$.
    Since both the derivative and the parser
    simply choose the branch of the {\sf Select}
    corresponding to $c$, this case is trivial.
  \end{itemize}
\end{proof}

\clearpage
\section{Full Experimental Results}\label{appendix:full-results}
\begin{figure}[H]
  \centering
  \includegraphics[width=.45\columnwidth]{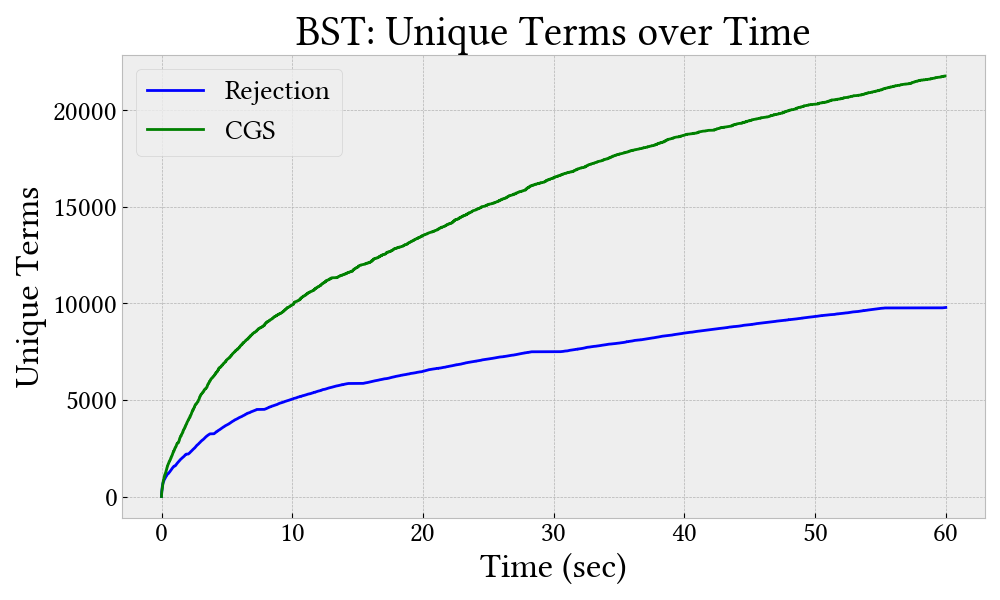}
  \includegraphics[width=.45\columnwidth]{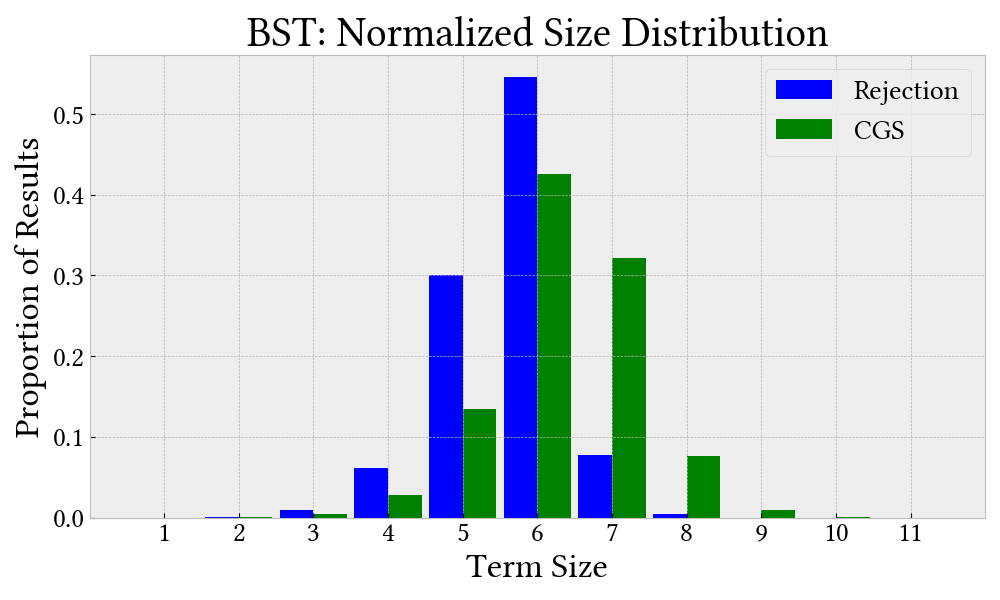}
  \caption*{\bench{BST} Charts}
\end{figure}

\begin{figure}[H]
  \includegraphics[width=.45\columnwidth]{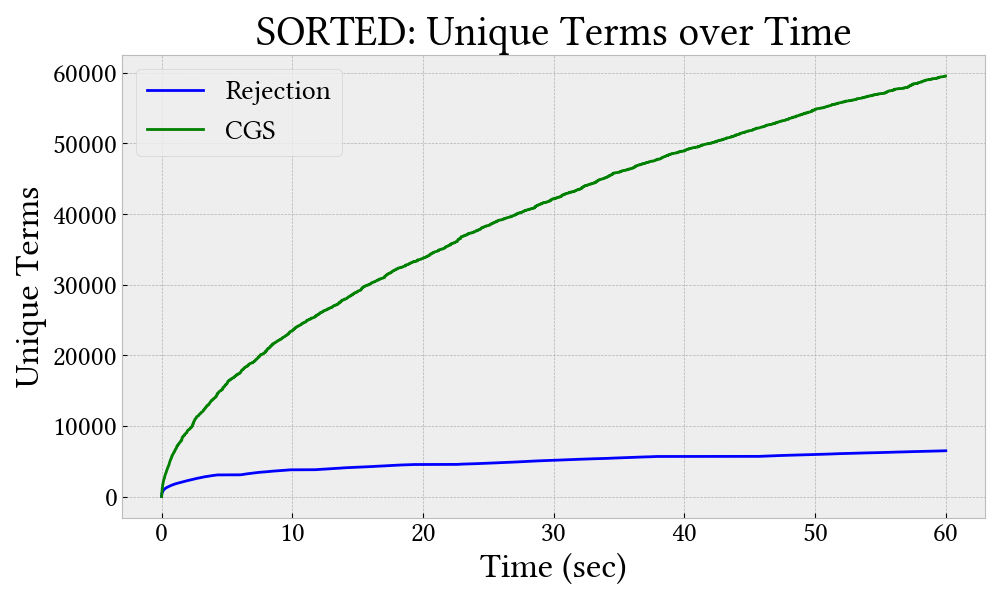}
  \includegraphics[width=.45\columnwidth]{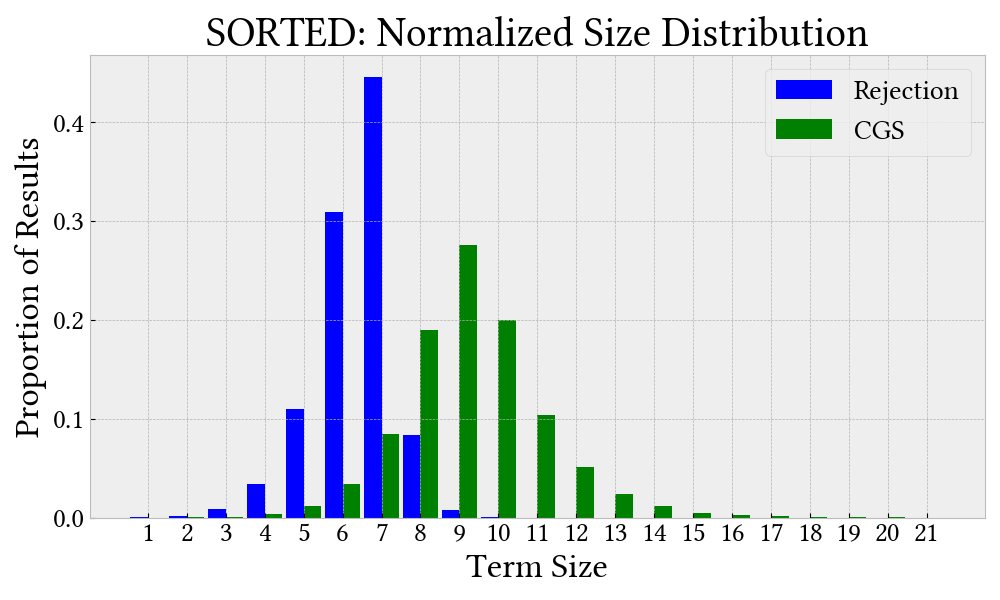}
  \caption*{\bench{SORTED} Charts}
\end{figure}

\begin{figure}[H]
  \centering
  \includegraphics[width=.45\columnwidth]{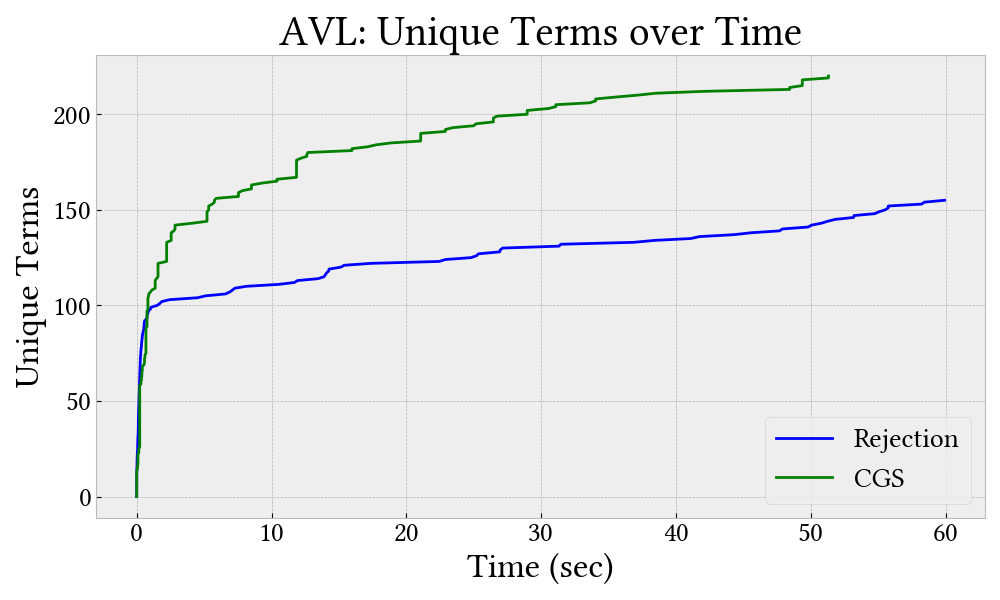}
  \includegraphics[width=.45\columnwidth]{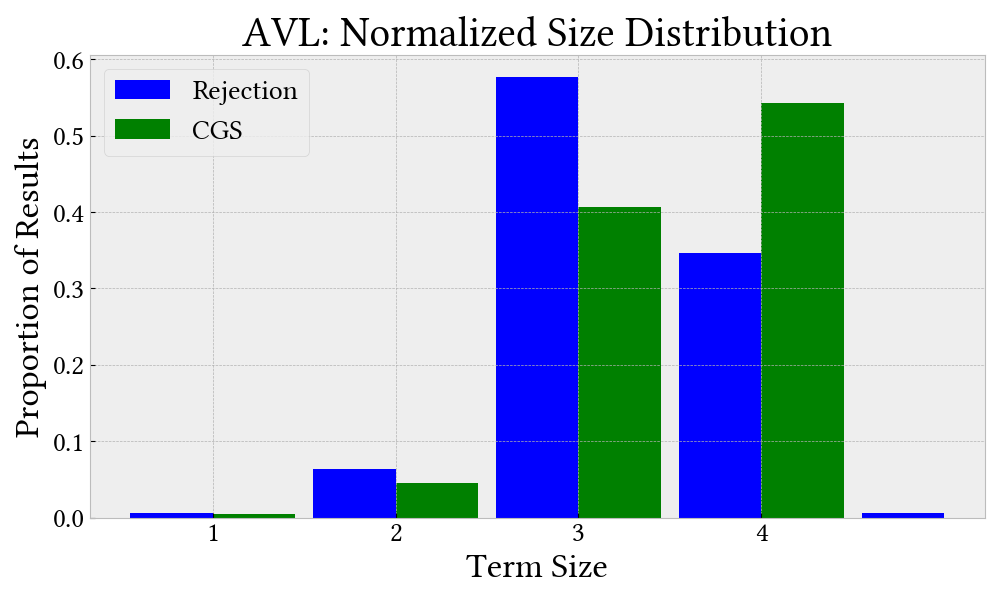}
  \caption*{\bench{AVL} Charts}
\end{figure}

\end{document}